%% file: main.tex
\documentclass[sigconf, 10 pt]{acmart}
\renewcommand\footnotetextcopyrightpermission[1]{} % removes footnote with conference information in first column
\usepackage{graphicx}
\usepackage{balance}  % for  \balance command ON LAST PAGE  (only there!)
\usepackage{tabularx, booktabs}
\usepackage{url}
\usepackage{amsmath}
\usepackage{dsfont}
\usepackage{algorithmicx, algpseudocode}
\usepackage[commentsnumbered,titlenumbered,linesnumbered,ruled,vlined]{algorithm2e}
\usepackage[skins]{tcolorbox}
\makeatother
\usepackage{graphicx}
\usepackage{subfig}
\usepackage{caption}
\usepackage{multirow}
\usepackage{epstopdf}
\usepackage{colortbl}
\usepackage{balance}
\usepackage{enumitem}
\usepackage{parallel,enumitem}
\usepackage{footnote}
\usepackage{multicol}
\makesavenoteenv{tabular}
\makesavenoteenv{table}
\usepackage{array} 
\newcolumntype{L}[1]{>{\raggedright\let\newline\\\arraybackslash\hspace{0pt}}m{#1}}
\newcolumntype{C}[1]{>{\centering\let\newline\\\arraybackslash\hspace{0pt}}m{#1}}
\newcolumntype{R}[1]{>{\raggedleft\let\newline\\\arraybackslash\hspace{0pt}}m{#1}}

\usepackage{pifont}% http://ctan.org/pkg/pifont
\newcommand{\cmark}{\ding{51}}%
\newcommand{\xmark}{\ding{55}}%

\usepackage{tikz}
\usetikzlibrary{decorations.pathreplacing,calc}
\newcommand{\tikzmark}[1]{\tikz[overlay,remember picture] \node (#1) {};}

\newtheorem{theorem}{Theorem}
\newtheorem{corollary}{Corollary}

\newtheorem{Lemma}{Lemma}
\newtheorem{lemma}{Lemma}
\newtheorem{Definition}{Definition}
\def\CCM{\textsf{CCM}}

\def\IM{\textsf{IM}}

\def\I{\mathbb{I}}
\def\RIS{\textsf{RIS}}

\def\SSA{\textsf{SSA}}
\def\Cov{Cov}
\def\Covw{\textbf{Cov}_w}

\def\OPT{\textsf{opt}}
\def\OPTk{\textsf{OPT}}
\def\eOPTk{\emph{\textsf{OPT}}}
\def\E{\mathbb{E}}

\def\tpa{\textsf{BCA}}

\def\dta{\textsf{DTA}}
\def\fs{\textsf{FS}}
\def\rs{\textsf{RS}}

\def\H{\mathcal{H}}
\def\E{\mathcal{E}}
\def\topk{\textsf{top}-k}

\def\IMM{\textsf{IMM}}
\def\DPIMA{\textsf{DSSA}}

\def\hedge{\textsf{Hedge}}
\def\yalg{\textsf{Y-alg}}

\newcommand*{\Resize}[2]{\resizebox{#1}{!}{$#2$}}%

\def\lc{\left\lceil}
\def\rc{\right\rceil}

\usepackage{mathtools}

\let\oldnl\nl% Store \nl in \oldnl
\newcommand{\nonl}{\renewcommand{\nl}{\let\nl\oldnl}}% Remove line number for one line

\newcommand{\twopartdef}[3]
{
	\left\{
	\begin{array}{ll}
		#1 & \mbox{if } #2 \\
		#3 & \mbox{otherwise}
	\end{array}
	\right.
}

\newcolumntype{Y}{>{\raggedleft\arraybackslash}X}
\setcopyright{none}
\begin{document}
\title{Approximate k-Cover in Hypergraphs: Efficient~Algorithms, and Applications}
%\titlenote{Produces the permission block, and
%  copyright information}
%\subtitle{Extended Abstract}
%\subtitlenote{The full version of the author's guide is available as
%  \texttt{acmart.pdf} document}

\author{Hung Nguyen}
\authornote{The work of Hung Nguyen was done while he was a Ph.D student at Virginia Commonwealth University}
\affiliation{%
 \institution{Carnegie Mellon University}
}
\email{hungnt@vcu.edu}
\email{hungnguy@andrew.cmu.edu}
\author{Phuc Thai}
\affiliation{%
 \institution{Virginia Commonwealth University}
}
\email{thaipd@vcu.edu}
\author{My Thai}
\affiliation{%
 \institution{University of Florida}
}
\email{mythai@cise.ufl.edu}
\author{Tam Vu}
\affiliation{%
 \institution{University of Colorado Boulder}
}
\email{tam.vu@colorado.edu}
\author{Thang Dinh}
\affiliation{%
 \institution{Virginia Commonwealth University}
}
\email{tndinh@vcu.edu}
%\authornote{The secretary disavows any knowledge of this author's actions.}
%\affiliation{%
%  \institution{Institute for Clarity in Documentation}
%  \streetaddress{P.O. Box 1212}
%\email{webmaster@marysville-ohio.com}
%
%\author{Lars Th{\o}rv{\"a}ld}
%\authornote{This author is the
%  one who did all the really hard work.}
%\affiliation{%
%  \institution{The Th{\o}rv{\"a}ld Group}
%  \streetaddress{1 Th{\o}rv{\"a}ld Circle}
%\email{larst@affiliation.org}
%\affiliation{%
%  \institution{Inria Paris-Rocquencourt}
%\affiliation{%
% \institution{Rajiv Gandhi University}
%\affiliation{%
%  \institution{Tsinghua University}
%  \streetaddress{30 Shuangqing Rd}
%  \country{China}
%\affiliation{%
%  \institution{Palmer Research Laboratories}
%  \streetaddress{8600 Datapoint Drive}
%\affiliation{\institution{The Th{\o}rv{\"a}ld Group}}
%\email{jsmith@affiliation.org}
%
%\author{Julius P.~Kumquat}
%\affiliation{\institution{The Kumquat Consortium}}
%\email{jpkumquat@consortium.net}
%
%% The default list of authors is too long for headers.

\begin{abstract}
Given a weighted hypergraph $\H(V, \E \subseteq 2^V, w)$, the approximate $k$-cover problem seeks for a size-$k$ subset of $V$ that has the maximum weighted coverage by \emph{sampling only a few hyperedges} in $\E$. The problem has emerged from several  network analysis applications including viral marketing, centrality maximization, and landmark selection. Despite many efforts, even the best approaches require $O(k n \log n)$ space complexities, thus, cannot scale to, nowadays, humongous networks without sacrificing  formal guarantees. In this paper, we propose \tpa{}, a family of algorithms for approximate $k$-cover that can find  $(1-\frac{1}{e} -\epsilon)$-approximation solutions within an \emph{$O(\epsilon^{-2}n \log n)$ space}. That is a factor $k$ reduction on space comparing to the state-of-the-art approaches with the same guarantee.  We further make \tpa{} more efficient and robust on real-world instances by introducing a novel adaptive sampling scheme, termed \dta.
%\textit{$k$-cover} on Hypergraphs generalizes many core problems in diverse application areas of network sciences, e.g. influence maximization, group centrality maximization. The $k$-cover finds a set of $k$ nodes that intersects the most hyperedges in the hypergraph. Existing research works found near-optimal solution but face a huge performance challenge as they resort to heuristic ad-hoc settings and lose the golden performance guarantees. In this work, we develop a family of simple, highly time and space-efficient algorithms that find $(1-\frac{1}{e} -\epsilon)$-approximate solution and use small amount of time and space. Extensions to weighted and budgeted problems and solving problem in more demanding streaming model are accommodated with minor changes. Our empirical 
Comprehensive evaluations on applications confirm consistent superiority of \dta{} over the state-of-the-art approaches. \dta{} reduces the sketch size up to 1000x   and runs 10x faster than the competitors while providing  solutions with the same theoretical guarantees and comparable quality.
\end{abstract}



% \keywords{ACM proceedings, \LaTeX, text tagging}

\maketitle

%body from here
\input{body/introduction}

\input{body/prelim}

\input{body/tpa}

\input{body/dynamic}
\input{body/stream}
\input{body/new_exp}

\section{Conclusion}
We introduce a simple, space and time-efficient algorithmic framework, called \tpa{}, that finds a near-optimal $(1-1/e-\epsilon)$-approximate solution for the fundamental $k$-cover problem.
% The algorithm builds upon a core component of a fixed threshold algorithm (\tpa{}) that takes a bound $z$ for capping the maximum coverage and produces a solution $S_z$ with quality within $1-(1-\frac{1}{k})^k$ of $z$.
Our algorithms only requires $O(\epsilon^{-2} n \log n)$-space thanks to reduced sketch feature, $O(\epsilon^{-2} k n \log n)$-time and can be easily extended to budgeted, weighted versions. The empirical results show a leap performance benefit in saving time and memory in various applications of $k$-cover in comparison existing algorithms.

\newpage
% \textbf{           }
% \newpage
\bibliographystyle{ACM-Reference-Format}
\bibliography{disk,distributed,infEst,pids,sampling,social,targetedIM,budgetedIM}

\input{body/appendix}
\input{body/extended}
\end{document}

%% file: body/introduction.tex
\section{Introduction}
\label{sec:intro}
Given a hypergraph $\H=(V, \E \subseteq 2^V)$, in which each hyperedge $E\in \E$ can contains one, two, or more nodes, we investigate the \emph{approximate~$k$-cover}  problem that seeks a subset of $k$ nodes in $V$ that cover maximum number of hyperedges. While approximate $k$-cover is equivalent to the classic Max-$k$-cover problem\footnote{ We can create an equivalent instance of Max-$k$-cover in which there is a subset $S_v$ for each nodes $v \in V$, an element $e$ for each hyperedges $E \in \E$, and $e \in S_v$ iff $v \in E$} \cite{Karp72} allowing a standard greedy $(1-1/e)$-approximation algorithm, the high space and time complexity  $O(|V| |\E|)$ of the greedy makes it intractable for large instances originated from real-world applications in network analysis. In those applications, including influence maximization \cite{Borgs14, Tang14, Tang15}, landmark selection \cite{Potamias09}, and centrality maximization \cite{Mahmoody16}, $\E$ often contains \emph{(exponentially) large number of hyperedges}, thus, we can only afford finding solutions via \emph{sampling a few hyperedges}. 

For example, the landmark selection (LMS) problem in \cite{Potamias09} aims to find $k$ landmarks (nodes) that lies on the maximum number of shortest paths in a graph $G=(V, E)$. The problem  can be transformed into an instance of approximate $k$-cover
on a hypergraph $\H'=(V'=V, \E')$ with  $\E' =\{ E_{s,t}| s, t \in V\}$ contains $O(n^2)$ hyperedges $E_{s, t}$ defined as  all vertices that are on some shortest path from $s$ to $t$. Finding an $(1-1/e)$ approximate solution for LMS takes $O(n^3)$ space and time and, thus, not tractable for  networks with millions of nodes. For other applications like influence maximization, the number of hyperedges, $O(n 2^n)$, limits the exact greedy to only toy networks.

State-of-the-art approaches for influence maximization \cite{Borgs14,Tang14,Tang15,Nguyen163} and centrality maximization \cite{Yoshida14,Mahmoody16}, despite different applications, all solve \emph{implicit} instances of approximate $k$-cover via a common framework, termed \emph{full sketch}.
The framework contains two phases: 1) a collection of many random hyperedges are generated using a random oracle, and, 2) greedy algorithm is invoked to solve max-$k$-cover over the generated hyperedges. 
To guarantee $(1-1/e-\epsilon)$-approximate solutions, even the best approaches incurs a \emph{high space complexity} $O( n k \log n)$.   Empirically, two recent works in \cite{Tang17,Arora17} indeed confirm that existing algorithms for influence maximization like \IMM{} \cite{Tang15} and \SSA/\DPIMA{} \cite{Nguyen163} only perform well on the \textsf{WC} weight setting but run out of memory for other weight settings, e.g., the trivalency model \cite{Arora17} or the slightly perturbed \textsf{WC} \cite{Tang17}.

In this paper, we introduce \tpa{}, a framework for efficient approximation algorithms for approximate $k$-cover. \tpa{} provides an $(1-1/e -\epsilon)$-approximation within $O( n \log n)$ space, a factor $k$ reduction comparing to the state-of-the-art methods. This space reduction is critical for large instances of approximate $k$-cover, especially, when running on memory-bounded systems. The key novelty of our \tpa{} framework is a coupling of a \emph{reduced sketch} and \emph{space-efficient bounding methods}. Our {reduced sketch} allows non-essential samples (hyperedges) to be removed from the sketch from time to time, thus, keeping a small memory footprint. This contrasts the full-sketch framework used in  existing approaches \cite{Tang14, Tang15, Tang17, Nguyen163,Mahmoody16} of which all generated samples needs to be kept on the memory in order to apply the standard greedy algorithm. The size of our sketch is reduced in order the magnitude when compare the full sketch, while the running time and the solution quality are  comparable to those in the full sketch.
%as shown in Figure~\ref{fig:mem_time_quality}. 
Our space-efficient bounding methods helps make smart decisions on when we  select the next node into the solution and reduce the sketch, accordingly.

% \begin{figure}[!ht]
% 	\vspace{-0.4in}
% 	\includegraphics[width=\linewidth]{figures/mem_time_quality}
% 	\vspace{-0.75in}
% 	\caption{\small Comparisons between reduced sketch and full sketch. The reduced sketch shrinks the sketch size by 5-100x times while providing comparable solution quality in comparable running time. More details are given in subsection \ref{subsec:effectinvess}.}
% 	\label{fig:mem_time_quality}
% %	\vspace{-0.05in}
% \end{figure}

We further reduce the space and time consumption of \tpa{} by introducing a novel adaptive sampling scheme, termed \dta. \dta{} can adapt well to the real complex of the real-world instances, providing the best overall approach across different applications of approximate $k$-cover without compromising the theoretical guarantee.

Our comprehensive experiments on various applications including k-dominating set, landmarks selection (aka coverage centrality maximization), and influence maximization, show the consistent superiority of our \tpa{} framework, particularly, \dta{} over the existing approaches. \dta{} uses substantially less time and memory than its state-of-the-art competitors, up to 10x faster and 1000x reduction in sketch size while providing comparable quality. Our algorithm can find good landmarks in large networks in minutes, and, remains the only algorithm for \IM{} that scale to billion-scale networks on challenging input settings.

%Moreover, a favorable feature of existing and our proposed methods is finding solution subject to a desired quality, e.g. $(1-1/e-\epsilon)\OPTk$ for input $\epsilon$. Our results on a broad range of parameter settings reveal that existing approaches are conservative while our solution provides a smooth trade-off between time and quality. Specifically, existing algorithms tend to return mostly the same solution quality for a wide range of the parameter $\epsilon$ that limit their practical uses. On the other hands, the solution quality returned by \tpa{} is nearly parallel with the desired $(1-1/e-\epsilon)\OPTk$ while reducing significantly time and memory requirements for larger $\epsilon$.
%With a large enough threshold value $z = z^*$, the \tpa{} algorithm returns an $(1-\frac{1}{e}-\epsilon)$-approximate solution for $k$-cover for any given $0 < \epsilon < 1-\frac{1}{e}$ which also controls $z^*$. We 

We summarize our contributions as follows:
\vspace{-0.1in}
\begin{itemize}
	\item We formally state the approximate $k$-cover problem and its connection to many important network analysis problems including influence maximization, landmark selection, and centrality maximization problems.
	\item We propose a family of $(1-1/e-\epsilon)$-approximation algorithm, called \tpa, that use only  $O(\epsilon^{-2} n \log n)$ spaces,  a factor $k$ reduction comparing to the state-of-the art approaches. 
	%An empirical approximation factor that potentially goes beyond $(1-1/e-\epsilon)$ is derived. 
	We provide natural extensions of \tpa{} to the budgeted problem. % and streaming in element arrival model.
	\item  We introduce a novel adaptive sampling scheme, termed \dta, that makes \tpa{} adapt to real-world instances. \dta{} consistently provides the fastest overall approach across different applications of approximate $k$-cver without compromising the theoretical guarantee.
	\item We carry comprehensive experiments on various network anlysis problems to demonstrate the consistent superiority of \tpa{} and its adaptive version \dta{} over the existing approaches, as well as to reveal the time-quality trade-off of each algorithm.
\end{itemize}

%% file: body/prelim.tex
\vspace{-0.08in}
\setlength{\textfloatsep}{4pt}
\section{Preliminaries}
\textbf{Hypergraph.} Let $\H=(V, \E)$ be a hypergraph that consists of a set $V$ of $n$ nodes and a multiset $\E$ of $m$ hyperedges such that $\forall E \in \E, E \subseteq V$, i.e., each hyperedge can have zero, one, two, or more nodes. A node $v\in V$ is said to \emph{cover} a hyperedge $E \in \E$,  if $v \in E$. Equivalently, we say that such hyperedge $E$ is \emph{incident} to $v$. 

Given a multiset $\E'$ of hyperedges in $\E$, the \emph{coverage} of a node $v$, denoted by $\Cov(v, \E')$ is defined as the number of hyperedges in $\E'$ incident to $v$. Generally, we define the coverage of a subset of nodes $S \subseteq V$ over $\E'$ as the number of hyperedges in $\E'$ incident to at least one node in $S$, i.e.
%
%For a subset of nodes $S \subseteq V$, we say $S$ covers hyperedge $E_j$ if $S$ intersect $E_j$, i.e. $S \cap E_j \neq \emptyset$. Then, the coverage of $S$ on hypergraph $\H$, denoted by $\Cov(S, \H)$ or simply $\Cov(S)$, is defined as the number of hyperedges in $\H$ covered by $S$, i.e.
\begin{align}
	\Cov(S, \E') = \left | \left \{ E \in \E | S \cap E \neq \emptyset \right \}\right|.
\end{align}
\textbf{Weighted Hypergraph.} We consider the more general weighted hypergraph $\H = (V, \E, w)$ in which each hyperedge $E \in \E$ is associated with a weight $w_E$ and $\sum_{E \in \E} w_E = 1$, thus, $w$ can be seen as a probability distribution over hyperedges. The \emph{weighted coverage} of a set of nodes $S$ on the complete set $\E$ of all the weighted hyperedges, denoted and shortened by $\Covw(S)$, is defined as follows:
\begin{align}
	\Covw(S) = \sum_{E \in \E, E \cap S \neq \emptyset} w_E.
\end{align}

%The coverage of a node $v \in V$ is also called the coverage of $v$ on the set of hyperedges. 
%Then, the $k$-cover problem is formulated as follows.
\begin{Definition}[$k$-cover]
	\label{def:max_k_cover}
	Given a weighted hypergraph $\H=(V,\E,w)$ and an integer $k$, find a subset $S^* \subseteq V$ of size $k$ that has the maximum weighted coverage among all the size-$k$ subsets $S \subseteq V, |S| = k$, i.e.
	\begin{align}
		 \Covw(S^*) = \emph{\OPTk{}} = \max_{S \subseteq V, |S| = k} \Covw(S).
	\end{align}
	$S^*$ is called the optimal solution.
\end{Definition}

The $k$-cover is a typical NP-hard problem \cite{Nemhauser81}. Furthermore, the coverage function is monotone and submodular, hence, the standard Greedy algorithm finds an $(1-1/e)$-approximate solution \cite{Nemhauser81} and the approximation factor is tight under a widely accepted assumption \cite{Feige98}. However, it requires the presence of all the hyperedges and in most applications, this requirement is impractical.
\vspace{0.03in}

\noindent \textbf{Random Sample Oracle and Approximate $k$-cover.} We are interested in the cases when the multiset of hyperedges $\E$ is too large (in terms of memory and time) to be listed explicitly. However, there exists an \emph{oracle} to sample random hyperedges from $\E$ with probability of generating the hyperedge $E$ equal to its weight $w_E$. Let $E_j$ be a random hyperedge generated by the oracle, then
\begin{align}
	\label{eq:sample_prob}
	\Pr[E_j = E] = w_E.
\end{align}
Many $k$-cover instances originated from network analysis problems such as influence and centrality maximizations are, indeed, of exponentially large size and such random sample oracles for them are known (detailed descriptions in Subsection~\ref{sec:apps}). Then, the sample hyperedges generated by the oracle can be used to approximate the weighted coverage of any subset $S \subseteq V$ as shown below.

%Let $E_1, E_2, \dots, E_j, \dots$ be the random hyperedges generated by the oracle. We obtain the following lemma.
\begin{Lemma}
	\label{lem:sample_weight}
	Given a weighted hypergraph $\H = (V, \E, w)$, let $E_j$ be a random hyperedge generated by the random oracle satisfying Eq.~(\ref{eq:sample_prob}) and a subset of nodes $S \subseteq V$. Let
	\begin{align}
		\label{xsj}
		X^{(j)}_S = \twopartdef{1}{E_j \cap S \neq \emptyset}{0}
	\end{align}
	be a Bernoulli random variable defined on $E_j$. Then, we have
	\begin{align}
		\mathbb{E}[X^{(j)}_S] = \mu_S = \Covw(S).
	\end{align}
\end{Lemma}
\begin{proof}

Let $\mathds{1}_{E \cap S \neq \emptyset} = 1$ if $E \cap S \neq \emptyset$ and $0$ otherwise.
\begin{align}
\mathbb{E}[X^{(j)}_S] & = \Pr[E_j \cap S \neq \emptyset] = \sum_{E \in \E} \mathds{1}_{E \cap S \neq \emptyset} \cdot \Pr[E_j = E] \nonumber \\
& = \sum_{E \in \E} \mathds{1}_{E \cap S \neq \emptyset} \cdot w_E = \sum_{E \in \E, E \cap S \neq \emptyset} w_E = \Covw(S). \nonumber
\end{align}
\end{proof}
Then, given a set $\E'$ of random hyperedges generated by the oracle, one can obtain an approximate $\Cov(S, \E')/|\E'|$ of $\Covw(S)$ for any set $S \subseteq V$. Hence, the $k$-cover, i.e. maximizing $\Covw(S)$, can be solved approximately through finding a set $\hat S$ of $k$ nodes to cover the most generated hyperedges in $\E'$, i.e. maximizing $\Cov(S, \E')$. Meanwhile, we also desire to find a provably good solution to the original $k$-cover problem. This gives rise to an approximate version of $k$-cover defined as follows:

%\textbf{Element-arrival Model.} We aim to design algorithms that can provide quality solutions within sub-linear time. We assume that the set of nodes $V$ is known in advance, however, the hyperedges arrive one-by-one, e.g. via sampling or streaming mechanisms. The hyperedges are processed on-the-fly as they arrive. The major challenge is to optimize the trade-off between, on one hand, the sample complexity and space complexity, and, on the other hand, the theoretical and empirical quality of the found solutions.

\begin{Definition}[Approximate $k$-cover]
	\label{def:approx_k_cover}
	Given a hypergraph $\H=(V,\E, w)$, a random oracle satisfying Eq.~(\ref{eq:sample_prob}), an integer $k$, an approximation factor $\rho \in (0, 1)$, find a subset $\hat S \subseteq V$ of size $k$ based on samples generated by the oracle that has its weighted coverage at least $\rho$ times the optimal with a high probability of $(1-\delta)$ for a small $\delta < 1$, e.g., $\delta = 1/n$,
	\begin{align}
		\Pr[ \Covw(\hat S) \geq \rho \emph{\OPTk} ] \geq 1-\delta.
	\end{align}
\end{Definition}
The goal is to identify such an approximate solution $\hat S$ in space and time \emph{independent of the size of $\E$}. The commonly desired approximation ratio is $\rho = 1-1/e-\epsilon$ provided a constant $\epsilon$. Many works in different applications of $k$-cover \cite{Yoshida14,Mahmoody16,Tang14,Tang15,Nguyen163} attempt to generate just enough random hyperedges (deciding memory usage and running time) to provide an $(1-1/e-\epsilon)$ approximate solution. Thresholds on that number of hyperedges are proposed and work in certain settings. However, space and time complexities are still high that stops them from running on very large networks as demonstrated in our experiments (Section~\ref{sec:exps}).

% Table~\ref{tab:syms} summarizes the common notations used in the paper.

\subsection{Applications and Related Work}
\label{sec:apps}

Here we present a variety of network analysis  applications that reduce to approximate $k$-cover, their corresponding hyperedge sampling oracles, and the related works. 
%We summarize the applications in Table~\ref{tab:apps}.

%\renewcommand{\arraystretch}{0.8}
%\setlength\tabcolsep{4pt}
%\begin{table}[!h]
%	%	\vspace{-0.1in}
%	\caption{$k$-cover applications and sampling algorithms.}
%	\vspace{-0.1in}
%	\label{tab:apps}
%	\centering
%	\begin{tabular}{L{2cm}L{1.7cm}L{3.9cm}}
%		\toprule
%		\multicolumn{2}{l}{\textbf{Application}} & \textbf{Sampling Algorithm}\\
%		\midrule
%		\multicolumn{2}{l}{Other from the google distributed paper} & test\\
%		\midrule
%		\multicolumn{2}{l}{\textsf{I/O Efficient Representatives}} & Random access of hyperedges\\
%		\midrule
%		\multicolumn{2}{l}{\textsf{Influence Maximization}} & \RIS{} sampling \cite{Borgs14} \\
%		\midrule
%		\multirow{5}{*}{\textsf{Centrality Max.}} & Betweenness & Sample shortest-path \cite{Mahmoody16}\\
%		\cmidrule{2-3}
%		& Coverage & Sample shortest-paths \cite{Mahmoody16}\\
%		\cmidrule{2-3}
%		& Triangle & Sample triangle \cite{Mahmoody16}\\
%		\cmidrule{2-3}
%		& $K$-Path & Sample $K$-Path \cite{Mahmoody16}\\
%		\bottomrule
%	\end{tabular}
%\end{table}

\renewcommand{\arraystretch}{1.2}
\setlength\tabcolsep{5pt}
\begin{savenotes}
	\begin{table*}[!h] \small
		%\vspace{-0.1in}
		\caption{Summary of algorithms for $k$-cover applications. The time to sample hyperedges are excluded. $\tilde{O}()$ suppresses the \textsf{polylog} factors from the complexities. Adaptivity means the ability to adapt to actual complexity of real-world instances.}
		\vspace{-0.12in}
		\label{tab:algorithms}
		\centering
		\begin{tabular}{L{2cm}L{4cm}|C{3cm}L{1.5cm}L{1.5cm}C{1.8cm}C{1cm}}
			\toprule
			\textbf{Algorithm} & \textbf{Problem} & \textbf{Guarantee} & \textbf{Space} & \textbf{Time} & \textbf{Adaptivity} & \textbf{Prefix} %& Sample complexity 
			\\
			\midrule
			\hedge{}\cite{Mahmoody16} 
			& k-Dominating set & $1-\frac{1}{e}-\epsilon$ \emph{only~if~$\OPTk=\Omega(m)$} & $\tilde{O}(\epsilon^{-2} k n)$ & - & \xmark{} & \xmark{} %& $O(\epsilon^{-2} k \log n \frac{m}{\OPTk})$ 
			\\
			\midrule
			\yalg{} \cite{Yoshida14} &  \CCM{}(LANDMARKS$_d$\cite{Potamias09}) & $(1-\frac{1}{e})\OPTk - \epsilon n^2$ & $\tilde{O}(\epsilon^{-2} n^2)$ & $\tilde{O}(\epsilon^{-2} n^2)$ & \xmark{} & \xmark{} %& $O(\epsilon^{-2} \log n)$
			\\
%			\midrule
%			\textsf{Stream}-$H_{\leq n}$ \cite{Bateni17} & Streaming edge-arrival $k$-cover & $1-\frac{1}{e}-\epsilon$ & $\tilde{O}(\epsilon^{-3} n)$ & - & \xmark{} %& \_ 
%			\\
%			\midrule
%			\textsf{Stream}-AH \cite{Mcgregor17} & Streaming set-arrival $k$-cover & $1-\frac{1}{e}-\epsilon$ & $\tilde{O}(\epsilon^{-2} n)$ & - & \xmark{} %& \_ 
%			\\
			\midrule
			\IMM{}\cite{Tang15} & Influence Maximization & $1-\frac{1}{e}-\epsilon$ & $\tilde{O}(\epsilon^{-2} k n)$ & $\tilde{O}(\epsilon^{-2} k n)$& \xmark{} & \xmark{} %& $O(\epsilon^{-2} k \log n \frac{m}{\OPTk})$
			\\			
			\DPIMA{}\cite{Nguyen163}& Influence Maximization & $1-\frac{1}{e}-\epsilon$ & $\tilde{O}(\epsilon^{-2} k n)$ & $\tilde{O}(\epsilon^{-2} k n)$& \cmark{} & \xmark{} %& $O(\epsilon^{-2} k \log n \frac{m}{\OPTk})$
			\\			
			\midrule
			Distributed sketch\cite{Bateni18} & $k$-cover & $1-\frac{1}{e}-\epsilon$ & $\mathbf{\tilde{O}(\epsilon^{-2} n)}$ /machine& - & \xmark{} & \xmark{} 
			\\
			\midrule
			\tpa{}($z^*$)[Here] & All above & $1-\frac{1}{e}-\epsilon$ & $\mathbf{\tilde{O}(\epsilon^{-2} n)}$ & $\tilde{O}(\epsilon^{-2} k n)$& \xmark{} & \cmark{} 
			%& $O(\epsilon^{-2} k \log n \frac{m}{\OPTk})$
			\\
			\dta{}[Here] & All above & $1-\frac{1}{e}-\epsilon$ & $\mathbf{\tilde{O}(\epsilon^{-2} n)}$ & $\tilde{O}(\epsilon^{-2} k n)$& \cmark{} & \xmark{} 
			%& $O(\epsilon^{-2} k \log n \frac{m}{\OPTk})$
			\\
			\bottomrule
		\end{tabular}
		\vspace{-0.1in}
	\end{table*}
\end{savenotes}

%Most external memory support random access model which allows accessing to any data on disks. 
%Since querying data is slow and the data is huge, the property of acquiring only a fraction of hyperedges of our proposed algorithms brings a huge benefit compared to other methods that require scanning through the whole input (possibly many times) to make decision, e.g. greedy algorithm \cite{Nemhauser81}.

%Applying our \dta{} algorithm gives the same results as in the original context of $k$-cover on hypergraphs.
%For example, popular online social networks like Facebook, Twitter hosts billions of online users and require Terabytes to store a single copy.
\vspace{0.03in}

\noindent \textbf{Landmark Selection.} To find good landmarks for shortest path queries, Potamias et al. \cite{Potamias09} solve a coverage centrality maximization (\CCM) problem.
Define combined coverage centrality of a group of nodes $S \subseteq V$ as follows:
\begin{align}
\mathbb{C}(S) = \sum_{s \neq t \in V} I_{s,t}(S),
\end{align}
where $I_{s,t}(S) = 1$ if at least one shortest path between $s$ and $t$ passes through some node in $S$.

Finding good landmarks is done by finding $k$ nodes $S$ that maximize the combined coverage centrality $\mathbb{C}(S)$. The greedy algorithm for LANDMARKS$_d$ in \cite{Potamias09} has a time complexity $O(n^3)$, thus, is not efficient for large networks. Our proposed approach here can be used to find good landmarks in much less time and space.

The corresponding $k$-cover for \CCM{} is defined on weighted hypergraph $\H = (V, \E, w)$ where $\E$ contains all hyperedges $R_{s\neq t \in V} = \{ u \in V | u \text{ is on at least one } s-t \text{ shortest path} \}$ with the same weight $w_{s-t} = 1/(n^2 -n)$ for all of them. Thus, $\H$ can be seen as an unweighted hypergraph and $\mathbb{C}(S) = (n^2 -n) \cdot \Covw(S)$.

\emph{Sample Oracle.} To generate a random hyperedge, we pick a random node pair $s \neq t$ from $V$ and return a set of nodes that are on at least one of $s-t$ shortest paths. Apparently, the probability of generating $R_{s\neq t \in V}$ is $w_{s-t}$ due to random selection of $s,t$ pair.

%\noindent where $\sigma_{s,t}(S)$ is defined similarly as for betweenness centrality.
%
%The maximization problem for coverage centrality is an instance of $k$-cover by the following procedure to generate a random sample for computing the coverage centrality \cite{Mahmoody16}:
%\begin{itemize}
%	\item[1)] Select a random pair of nodes $u,v \in V$.
%	\item[2)] Add to the sample all the nodes on all the shortest paths connecting $u$ and $v$.
%\end{itemize} 
%Each of the sample defines a hyperedges and there are $n(n-1)/2$ such hyperedges.

\vspace{0.03in}

\noindent \textbf{Influence Maximization (\IM{}).} Given a graph $\mathcal{G} = (V,E,p)$ where edge $(u,v) \in E$ has probability $p(u,v) \in (0,1]$ representing the influence of $u$ towards $v$ under some diffusion model, the influence of a set $S \subseteq V$, denoted by $\I(S)$, is the expected number of nodes who are eventually influenced by $S$.
%Here a diffusion model specifies a stochastic cascading process from $S$, that determines when a node $v \in V$ gets influenced. 
%For example, in the popular Independent Cascade (IC) model \cite{Kempe03}, an initially or newly influenced node $u$ has a single chance to influence each of his neighbor $v$ and succeeds with probability $w(u,v)$.
The Influence Maximization (\IM{}) problem \cite{Domingos01, Kempe03} seeks a set $S_k$ of at most $k$ nodes with the maximum influence. 
%The \IM{} problem is NP-hard \cite{Kempe03}.

\IM{} can be seen as an $k$-cover instance with $\H=(V, \E, w)$ where $\E$ contains all \emph{reverse influence sets} (\RIS{}s) \cite{Borgs14}. A random \RIS{} $R_j$ associates with a weight $w_j = \frac{1}{n} Pr_g, $ where $Pr_g$ is the probability that a sample graph $g$ is generated from $\mathcal{G}$ following the diffusion model. Then, the influence of a set of nodes $S$ is proportional to the weighted coverage of $S$ in $\H$, i.e. $\I(S) = n \cdot \Covw(S)$ \cite{Kempe03}, hence, maximizing $\I(S)$ is equivalent to maximizing $\Covw(S)$.

The size of $\E$ is exponentially large in most stochastic diffusion model. For example, in the popular independent cascade model \cite{Kempe03}, $|\E| = n 2^{|E|}$.

\emph{Sample Oracle}. A random \RIS{} $R_j \in \E$ is generated by 1) selecting a random node $u \in V$ 2) including into $R_j$ all the nodes that are reachable to $u$ in a sample graph $g$ generated according to a given diffusion model \cite{Borgs14}. Then, the probability of generating an \RIS{} $R_j$ is precisely $\frac{1}{n} Pr_g = w_j$ \cite{Borgs14}. This \RIS{} sampling algorithm serves as the oracle in our Approximate $k$-cover instance for \IM{}.

\emph{Related work.} \IM{} \cite{Leskovec07,Tang14,Tang15,Nguyen163,Ohsaka16,Du13} recently emerged from the vital application of designing a marketing strategy to maximize the benefit. Kempe et al. \cite{Kempe03} proved the submodularity property of influence function, thus, greedy algorithm provides an $(1-1/e-\epsilon)$-approximation guarantee. However, the naive greedy algorithm is not efficient due to the \#P-hardness of computing influences \cite{Chen10}. Lazy greedy technique was studied in \cite{Leskovec07}.
%Lazy greedy technique \cite{7} was investigated and shown to vastly improve the performance of naive greedy. 
Popular algorithms using \RIS{} sampling \cite{Borgs14}, are summarized in Table~\ref{tab:algorithms}. For a more comprehensive literature review on \IM{}, see \cite{Arora17}.
%Later, reverse influence sampling (\RIS) \cite{Borgs14} was introduced and \RIS{}-based algorithms \cite{Tang14,Tang15,Nguyen16,Nguyen163} were able to scale up to networks with millions of nodes and billions of edges. More practical applications of \IM{} were examined in \cite{Goyal113,Chen2015,Li15,Li152,Lu15,Ohsaka16,Du13}.
\vspace{0.03in}

\textbf{$k$-Dominating Set \cite{Bateni18}}. In this problem, we are given a graph $G(V, E)$ and we aim to find $k$ nodes that cover maximum number of (multi-hop) neighbors.

The $k$-dominating set is a special case of $k$-cover $\H=(V, \E)$ with $\E = \{ v \in V | E_v \}$ in which hyperedge $E_v$ contains all nodes in $V$ that can cover $v$ (via multi-hop neighborhoods).

The sample oracle would sample a random node $v \in V$ and generate via reverse reachability from $v$ all nodes that can cover $v$ \cite{Bateni18}.

%Sampling algorithms to generate random shortest-path, triangle and $K$-path as samples are proposed accordingly \cite{Mahmoody16} and, thus, the maximization problems can be seen as instances of the $k$-cover where $\E$ is the set of all possible samples.

\textbf{$k$-cover and Large-scale Submodular Maximization.} Many techniques have been developed for large-scale submodular maximization (see \cite{Mcgregor17,Bateni18} and the reference therein). However, many of the existing techniques only achieve sub-optimal approximation guarantees and/or space complexities and/or do not scale in practice. 

A recent work by Bateni et al. in KDD'18 \cite{Bateni18} provided the first algorithm for $k$-cover, a special case of submodular maximization, with \emph{optimal approximation guarantee}, $1-1/e-\epsilon$, and \emph{optimal space complexity} $\tilde{O}(n)$ per machine. We achieve the same optimal approximation guarantee and space complexity requiring only a \emph{single machine}. Further, the empirical results in \cite{Bateni18} had to resort to empirical method to determine parameters (e.g., $\sigma$), thus, it is not clear if the proposed algorithm still provide the worst-case guarantee $1-1/e-\epsilon$. In contrast, our proposed algorithms require only parameter $\epsilon$ on the approximation guarantee and can automatically decide the number of samples and other parameters. Yet our algorithm scale to very large instances of billion-scale.

%% file: body/tpa.tex
\renewcommand{\arraystretch}{1.2}
\setlength\tabcolsep{3pt}
\begin{table}[hbt]\small
	%	\vspace{-0.1in}
	\centering
	\caption{Summary of notations}
	\vspace{-0.15in}
	\begin{tabular}{m{1.8cm}|m{6.3cm}}
		\addlinespace
		\toprule
		\bf Notations  &  \quad \quad \bf Descriptions \\
		\midrule 
		$\mathcal{H} = (V, \E,w)$ & Weighted hypergraph of $n$ nodes, $m$ hyperedges\\
		\hline
		$\E_r$ & Reduced sketch maintained in \tpa{}\\
		\hline
		$\E_f$ & Full sketch of all generated hyperedges in \tpa{}\\
%		$\E_r^{(j)}$ & $\E_r$ after receiving hyperedge $E_j$ and before any update\\
		%		\hline
		%		$\Cov(S)$ & Hypercoverage of set $S$ in the hypergraph $\mathcal{H}$\\
		\hline
		$\Cov(S, \E')$ & Coverage of a set $S$ on $\E'$ of hyperedges\\
		\hline
		$\Covw(S)$ & Weighted coverage of a set $S$ on $\E$\\
		\hline
		$S^*, \OPTk$ & The true optimal solution of $k$ nodes for $k$-cover and its optimal weighted coverage value\\
		\hline
		$\OPT$ & Running maximum coverage of $k$ nodes on a set of generated hyperedges so far in \tpa{} algorithm.\\
		\hline
		$f(S,d_S,\E_r)$ & Upper-bound function on the maximum coverage\\
		\hline
		$f_L(n,\delta,\hat \mu, N)$ & Lower-bound function of $\mu$ with probability $1-\delta$ given an estimate $\hat \mu$ on $n$ hyperedges\\
		\hline
		$f_U(n,\delta,\hat \mu, N)$ & Upper-bound function of $\mu$ \\
		\bottomrule
	\end{tabular}
	\label{tab:syms}
\end{table}
\section{Coverage Thresholding Algorithm}
\label{sec:tpa}
We propose a family of simple, efficient \tpa{} algorithms for $k$-cover problem. It guarantees to find an $(1-1/e-\epsilon)$-approximate solution with probability $1-\delta$ for any given $0 < \epsilon < 1-1/e, 0 < \delta < 1$. Moreover, it is easily extendable to the budgeted version of the problem with minor changes (See extensions in Section~\ref{sec:extensions}). We summarize the notations in Table~\ref{tab:syms}.
%We provide comparison with existing approaches from different applications in Table~\ref{tab:algorithms}.

\subsection{Reduced Sketch \& A Family of Bounded Coverage Algorithms (BCA)}

Our family of Bounded Coverage Algorithms (\tpa{}) is presented in Algorithm~\ref{alg:tpa} which: 
\begin{itemize}
	\item uses the random oracle to generate random hyperedges $E_1, E_2, \dots, E_j,\ldots$ (Line 6).
% Let $$\E^{(j)} = \{E_1, E_2, \ldots, E_j\}$$ be the set of the first $j$ random hyperedges generated. Let $\OPT^{(j)}$ be the maximum coverage of $k$ nodes on $\E_{j}$ where $k$ is an input of the number of nodes selected. 	
	\item receives a threshold $z$ bounding the maximum coverage of any set of $k$ nodes on generated hyperedges. 
	\item gradually constructs a partial solution $S$ and maintains its coverage $d_S$ on the full sketch (Lines~8,11,12).
	\item only maintains hyperedges that are not covered by the partial solution $S$ in a \emph{reduced sketch} $\E_r$ (Line~10).
\end{itemize}
\begin{algorithm} \small
	\caption{\tpa{} - Bounded Coverage Algorithms}
	\label{alg:tpa}
	\KwIn{A budget $k$, a threshold $z$ and a random oracle}
	\KwOut{ $T$ - maximum number of generated hyperedges,\\ 				  \qquad \qquad $S$ - a candidate solution, and \\
		\qquad \qquad $d_S$ - coverage of $S$}
	$S \leftarrow \emptyset$; \Comment{Gradually growing solution set}\\ 
	$d_{S} \leftarrow 0; $\Comment{Coverage of $S$ on all the generated hyperedges} \\
	$\E_r \leftarrow \emptyset$; \Comment{Reduced sketch: hyperedges not covered by $S$}\\
	%$j \leftarrow 1$;\\
	\For{$i=1$ \emph{\textbf{to}} $k$}{
		%    	\While {$\displaystyle\max_{v\notin S}\Cov (v,\mathcal E) <\frac{z - d_S}{k}$}{
		\While{$f(S,d_S,\E_r) < z$}{
			Generate a new hyperedge $E_j$;\\
			\If{$E_j \cap S \neq \emptyset$}{
				$d_S = d_S+1$ and discard $E_j$;\\
			} \uElse{
				Add $E_j$ to $\E_r$;\\
			}
		}
		Add $u=\arg\max_{v\notin S} \Cov (v,\E_r)$ to $S$;\\
		Remove hyperedges covered by $u$ from $\E_r$, update $d_S$;\\ 
	}
	\textbf{return} $<T = \text{size of full sketch}, S, d_S>$\\
\end{algorithm}

\tpa{} goes over $k$ iterations to select $k$ nodes into $S$ (Lines~4-12). In iteration $i$, it keeps generating new hyperedges as long as the threshold $z$ holds as a bound on the maximum coverage of $k$ nodes on the generated hyperedges so far. To enforce this condition on threshold $z$, we use a function $f(S, d_S, \E_r)$ to compute an upper-bound of the maximum coverage on the set of all generated hyperedges. Then, $z$ holds as long as $f(S, d_S, \E_r) < z$ (Line~5). When the threshold $z$ fails the check $f(S, d_S, \E_r) < z$, the node $u$ with maximum coverage on the reduced sketch $\E_r$ is inserted to the partial solution $S$ (Line~11). As a results, hyperedges covered $u$ are removed from reduced sketch $\E_r$ and the coverage $d_S$ of $S$ increases accordingly (Line~12). If a hyperedge $E_j$ is covered by the partial solution $S$, \tpa{} simply increases its coverage $d_S$ by 1 and discards $E_j$ (Lines~7,8). Otherwise, $E_j$ is added into the reduced sketch $\E_r$ (Line~10).

\renewcommand{\arraystretch}{1.2}
\setlength\tabcolsep{1.2pt}
\begin{table}[!h] \small
	%	\vspace{-0.1in}
	\caption{A running example of \tpa$(k=2,z=4)$ on a toy network of 3 nodes $\{1,2,3\}$.}
	\vspace{-0.1in}
	\label{tab:bca_example}
	\centering
	\begin{tabular}{c | c c c c c}
		\toprule
		Event & $\E_r$ & $S$ & $d_S$ & $f(S, d_s, \E_r)$ & Notes\\
		\midrule
		\_ & $\emptyset$ & $\emptyset$ & 0 & 0 & Initialization\\
%		\hline
		$E_1 = \{1,2\}$ & $\{E_1\}$ & $\emptyset$ & 0 & 2 & $f(S, d_s, \E_r) < z$\\
%		& $\{E_1\}$ & $\emptyset$ & 0 & 2 & No node is selected\\
%		\hline
		$E_2 = \{1,3\}$ & $\{E_1, E_2\}$ & $\emptyset$ & 0 & 4 & $f(S, d_s, \E_r^) = z$\\
		\textbf{Select node 1} & \textbf{$\emptyset$} & \textbf{$\{1\}$} & \textbf{2} & \textbf{2} & \textbf{Remove $E_1, E_2$ from $\E_r$}\\
%		\hline
		$E_3 = \{2\}$ & $\{E_3\}$ & $\{1\}$ & 2 & 4 & $f(S, d_s, \E_r) = z$\\
		\textbf{Select node 2} & \textbf{$\emptyset$} & \textbf{$\{1,2\}$} & \textbf{3} & \textbf{3} & \textbf{Remove $E_3$ from $\E_r$}\\
		\bottomrule
	\end{tabular}
\end{table}
\vspace{0.03in}

\textbf{Example.} Let $f(S, d_S, \E_r) = d_S + k \cdot \max_{v \notin S} \Cov(v,\E_r)$ be the upper-bound function (details in our analysis). We give a running example in Table~\ref{tab:bca_example} of $\tpa{}(k=2,z=4)$ on an instance of 3 nodes $\{1,2,3\}$. 
%The hyperedges observed by \tpa{} are labeled $E_1, E_2, \dots$ 
Let assume an instantiation of generated hyperedges as given in Table~\ref{tab:bca_example}. After initialization of $\E_r, S, d_S$, \tpa{} generates $E_1 = \{1,2\}$ and put in the reduced sketched since $f(S,d_S,\E_r) = 2 < z$. When $E_2 = \{1,3\}$ is generated, the upper-bound function $f(S,d_S,\E_r) = 4 = z$, triggering the selection of the node with maximum coverage, i.e. node 1, and both hyperedges $E_1, E_2$ are removed from $\E_r$ since they are covered by node 1. After selecting node 1, $f(S,d_S,\E_r) = 2 < z$ and the next hyperedge $E_3 = \{2\}$ is generated causing $f(S,d_S,\E_r) = 4 = z$ and node 2 having the maximum coverage of 1 is selected. \tpa{} stops as $|S| = 2$.
\vspace{0.03in}

\textbf{Upper-bound function.} The $f(S, d_S, \E_r)$ function provides a trade-off between solution quality and time/memory usage. On one hand, the tighter bounds provided by $f(S, d_S, \E_r)$ delay the selection of nodes in \tpa{}. That intuitively leads to a better solution. On the other hand, tighter bounds often require more time and space to find. We will present 5 different upper-bound functions with varying quality and time/memory relation.
\vspace{0.03in}

\textbf{Reduced sketch $\E_r$.} By only maintaining a reduced sketch $\E_r$ of uncovered hyperedges, \tpa{} has a strong advantage of throwing out (considerably many) unnecessary hyperedges reducing (considerable) memory space, up to 100x in practice (see Table~\ref{tab:rtime}). This is in stark contrast to existing approaches, e.g. \hedge{}\cite{Mahmoody16}, \yalg{}\cite{Yoshida14}, \IMM{}\cite{Tang15}, \DPIMA{}\cite{Nguyen163}, that have to store all generated hyperedges.

\textbf{Maintaining reduced sketch.} To guarantee, linear-time update for our reduced sketch, we store our sketch using an \emph{incidence graph} and apply a  \emph{lazy removal} strategy. 

 The incidence graph maintains mappings between nodes in $V$ and hyperedges in $\E_r$ as well as the reversed index. Specifically, each nodes $u \in V$ has a list of incident hyperedges and vice versus, each hyperedge has a list of nodes contained in it. Adding hyperedges to the incidence graph is straightforward and takes only linear time.

\emph{Removing hyperedges.} When a new node $u$ is selected into the solution (Line~11, \tpa{}), we need to remove from the sketch all hyperedges that covered by $u$ (Line~12, \tpa{}). That is we need to 1) remove all the corresponding hyperedges from the incidence graph   and 2) remove those hyperedges from all the incident lists of all nodes that appear in those hyperedges. A naive approach will incur expensive cost. 

\emph{Lazy removal.} Upon removing a hyperedge, we simply mark that hyperedge ``deleted'' without actual removal of the hyperedges and updating the incident lists of the affected nodes. We, however, keep a count of how many entries in the incidence graph are removed, thus, knowing the ratio between the amount of ``trash'' contained in the incidence graph.

Whenever the fraction of ``trash'' in the incidence graph exceeds a predefined constant $\alpha$, set to be $1/3$ in our experiment,  we invoke a `cleaning' routine. The `cleaning' routine will simply reconstruct the incidence graph from scratch, removing all marked for removal hyperedges. This batch removal of hyperedges guarantees \emph{linear amortized time complexity} for hyperedges removal in terms of the hyperedges' size.

\vspace{-0.1in}
\subsection{Analysis for an arbitrary threshold $z$}
\textbf{Solution guarantee analysis.} We show that \tpa{} algorithm returns $S$ with coverage $d_S \geq (1-(1-\frac{1}{k})^k) \cdot z \geq (1-\frac{1}{e}) \cdot z$ for any given threshold $z$, providing `sufficiently tight' upper-bound function $f(S, d_S, \E_r)$.

Denoted by $\OPT$, the \emph{running maximum coverage}, defined as the maximum coverage of any $k$ nodes on the full sketch $\E_f$.  Here, the full sketch refers to the collection of all the hyperedges that have been generated. Clearly, $E_r$, $E_f$, and $\OPT$ varies during the run-time of the algorithm. Thus, before using $\E_r, \E_f$ and $\OPT$, we will specify the moment at which they are defined. 

Particularly, at the \emph{termination of \tpa{}} with $T$ generated hyperedges, we will use the notations $\E_r^{(T)}, \E_f^{(T)}$ and $\OPT^{(T)}$ for  $E_r$, $E_f$, and $\OPT$, respectively.

Let
\begin{align}
\label{eq:f_r_func}
f_r(S, d_S, \E_r) = d_S + k \cdot \max_{v \notin S} \Cov(v,\E_r),
\end{align} 
termed the \emph{requirement function}. We first show that $f_r(S, d_S, \E_r)$ provides an upper-bound for $\OPT$.
\begin{Lemma}
	\label{lem:exist_f}
After generating hyperedge $E_j$ (Line~6, Alg.~\ref{alg:tpa}), the requirement function $f_r(S, d_S, \E_r)$  satisfies
	\begin{align}
		 f_r(S, d_S, \E_r) \geq \OPT.
	\end{align}
\end{Lemma}
\begin{proof}
	Consider the moment after sampling any hyperedge $E_j$, the full sketch $\E_f$ is composed of two parts: the reduced sketch $\E_r$ and hyperedges that are covered by $S$, i.e., $\E_f \backslash \E_r$. Thus, $|\E_f \backslash \E_r| = d_S$. We have 
  $\Cov(S^*, \E_f \backslash \E_r) \leq |\E_f \backslash \E_r| \leq d_S$. Since $\Cov$ is a monotone and submodular function \cite{Mcgregor17}, thus, $\Cov(S^*, \E_r) \leq k \cdot \max_{v\notin S} \Cov(v,\E_r)$. Therefore,
	 \[
	\OPT = \Cov(S^*, \E_f \backslash \E_r)+\Cov(S^*, \E_r)  \leq d_S + k \cdot \max_{v \notin S} \Cov(v,\E_r).
	\]
\end{proof}
%TODO: define w.r.t. the set of the first $t$ hyperedges $E_1, E_2,\ldots,E_t$, $\Cov^{(t)}(S), \OPT_k^{(t)}$ and $S^{*(t)}$.

We show that for any upper-bound function $f(S,d_S,\E_r)$ that is, at least, as `tight' as the requirement function, then the coverage of the returned solution $S$ is also  a constant factor close to $\OPT^{(T)}$.

\begin{Lemma}
	\label{lem:anyz}
	Let $\langle T, S, d_S \rangle$ be the output of \tpa{} for an integer $z > 0$ and $\rho_k =  1 -\left(1-\frac{1}{k}\right)^k \geq 1-\frac{1}{e}$. For any function $f(S,d_S,\E_r)$ satisfying the condition $$\emph{\OPT} \leq f(S,d_S,\E_r) \leq f_r(S, d_S, \E_r)$$ for all the moments after generating $E_j, 1 \leq j \leq T$, then
    \begin{align}
    	\label{eq:anyz}
    	z \geq \emph{\OPT}^{(T)} \geq d_S \geq \rho_k z.
    \end{align}
\end{Lemma}
%The proof of $\Cov(S, \E^{(T)}) \geq \rho_k z$ bases on a key realization that:
%\begin{align}
%	d^+_i - d^-_i \geq \frac{z-d^-_i}{k},
%\end{align}
%where $d^-_i$ and $d^+_i$ refer to the value of $d_z$ right before and after selecting the $i^{th}$ node into $S$. The other half is proved by the integrality of coverages, i.e., $\Cov(S, \E^{(T)}), \OPT^{(T)}$ and $z$.
%observing that before acquiring the last hyperedge, $f(S,d_z,\E) < z$ is maintained, meaning $\OPT^{T-1}_k < z$ and adding one more hyperedge can at most increase $\OPT^{T-1}_k$ by 1. Thus, $\OPT^{T}_k \leq \OPT^{T-1}_k + 1 \leq z$. 
%The detailed proof is presented in the appendix.

\begin{proof}
The proof consists of two components.

$\bullet$ \textit{Proving $\Cov(S, \E_f^{(T)}) \geq \rho_k z$.}	Consider the moment when we add the $i^{th}$ node (Line 11, Alg.~\ref{alg:tpa}), denoted by $u_i$, into $S$ for $i=1..k$. Denote by $d_i^-$ and $d_i^+$ the values of $d_S$ \emph{right before} and \emph{right after} adding the node, respectively. We also define $d_0^+= 0$.

Since $d_S$ is non-decreasing,  we have $d_{i-1}^+ \leq d_{i}^- \leq d_{i}^+,\ \forall i=1..k$.
Consider the moment right before the selection of $u_i$,
\begin{align}
& f(S,d_S,\E_r) \geq z \nonumber \\
\Rightarrow \text{ } & d_S + k \cdot \Cov(u_i, \E_r) \geq z \text{ \ \ (due to the lemma condition)}\nonumber \\
\Rightarrow \text{ } & d_i^- + k \cdot (d_i^+ - d_i^-) \geq z, \text{ for } i=1..k \nonumber \\
\Rightarrow \text{ } & d_i^{+} \geq d_i^- + \frac{z - d_i^-}{k}, \text{ for } i=1..k
\end{align}
Therefore,
\begin{align}
	\nonumber z - d_i^{+} \leq z - \Big(d_i^- + \frac{z - d_i^-}{k}\Big) & = \Big(1-\frac{1}{k}\Big)\Big(z-d_{i}^-\Big) \nonumber \\
	\label{eq:important_deriv}
	& \leq \Big(1-\frac{1}{k}\Big)\Big(z-d_{i-1}^+\Big).
\end{align}

Apply the above inequality recursively, we yield
\begin{align}
\nonumber z - d_k^+ \leq \Big(1-\frac{1}{k}\Big)\Big(z-d_{k-1}^+\Big) \leq ... & \leq \Big(1-\frac{1}{k}\Big)^k\Big(z-d_{0}^+\Big) \nonumber \\
& = \Big(1-\frac{1}{k}\Big)^k z.
\end{align}

It follows that
\begin{align}
	\Cov(S, \E_f^{(T)}) \geq d^+_k \geq  \Big(1-\Big(1-\frac{1}{k}\Big)^k\Big) z = \rho_k z.
\end{align}

$\bullet$ \textit{Proving $z \geq \OPT^{(T)}$.} Let $E_{T}$ be the last hyperedge generated in Line 6.  Consider the moment right before generating $E_{T}$ (before Line 6). Then, we have
\begin{align}
	\OPT \leq f(S,d_S,\E_r) < z.
\end{align}
Note that both $\OPT$ and $z$ are integer, we obtain
\begin{align}
	\OPT^{(T)} \leq \OPT + 1 \leq z.
\end{align}
Since, adding the last hyperedge can only change $\OPT$ by at most one, we yield the proof.  
\end{proof}

%\textbf{TODO:}
%- Need to highlight the important part of the proof:
%(1)$d_i - d_{i-1} \geq xyz$  (2) The $opt$ does not exceed $z$ due to the integrality and the fact that each hyperedge can increase the $opt$ by at most one.
%
%- Discuss/compare with the offline guarantee ($z$ vs. $opt$). The empirical approximation ratio: $d_S/z \geq 1-1/e$
Compared to the bound of $(1-\frac{1}{e})\OPT^{(T)}$ returned by traditional greedy algorithm \cite{Nemhauser81} on $T$ generated hyperedges, the quality guarantee of \tpa{} is stronger since $z \geq \OPT^{(T)}$.

\textbf{The requirement function $d_S + k \max_{v \notin S} \Cov(v,\E_r)$.} The simplest upper-bound function that satisfies the condition in Lemma~\ref{lem:anyz} is $f(S,d_S,\E_r) = d_S + k \cdot \max_{v \notin S} \Cov(v,\E_r)$, namely \emph{requirement function}. We will analyze the complexity of \tpa{} using this function and will present other more sophisticated upper-bound functions with space and time complexity in Section~\ref{sec:extensions}.
\vspace{0.03in}

\textbf{Complexity analysis.} The space and time complexities of \tpa{} algorithm for any $z$ are shown in the following.
\begin{Lemma}
	\label{lem:dta_complexity}
	\tpa{} algorithm with any $z$ and the requirement function $f_r(S, d_S, \E_r)$ uses $O(\frac{z \cdot n}{k})$ space and $O(z \cdot n)$ time.
\end{Lemma}
\begin{proof}
	Consider two cases:
	\begin{itemize}
		\item The moment after generating any hyperedge $E_j$ for $j < T$, then $f(S,d_S,\E_r) < z$, thus, $\Cov(v, \E_r) < \frac{z-d_S}{k}$.
		\item The moment after generating $E_T$, then $f(S,d_S,\E_r) \geq z$. We see that right before adding the hyperedge $E_T$ that makes $f(S,d_S,\E_r) \geq z$, $\Cov(v, \E_r) < \frac{z-d_S}{k}$. Hence, by only generating one more hyperedge $E_T$, the coverage of any node can only increase by 1, i.e., $\forall v \in V, \Cov(v, \E_r) < \frac{z-d_S}{k} + 1$.
	\end{itemize}
	Thus, $\Cov(v, \E_r) < \frac{z-d_S}{k}+1  \leq \frac{z}{k} + 1$ at all time. Since there are $n$ nodes, the space is bounded by $(\frac{z}{k}+1)n = O(\frac{z \cdot n}{k})$. 
	%Furthermore, the upper-bound function computation requires at most $O(n)$ to store the node coverages.
	%\begin{Lemma}
	%	\label{lem:zspace}
	%	For $z = z^*$ and $f(S,d_S,\E) = d_S + k \cdot \max_{v\notin S} \Cov(v,\E)$, \tpa{} algorithm uses $O(\epsilon^{-2} n \log n)$-memory space.
	%\end{Lemma}
	
	The time complexity is linear to the total size of all the observed hyperedges which is $O(z \cdot n)$ due to the fact that the maximum coverage of any node on all the observed hyperedges never exceeds $z+1$ and there are $n$ nodes.
	%the total size of all hyperedges acquired is bounded by $O(z^* n)$. Note that  used and hence, the following lemma on the time complexity of \tpa{} holds.
	%\begin{Lemma}
	%	\label{lem:ztime}
	%	For $z = z^*$ and $f(S,d_S,\E) = d_S + k \cdot \max_{v\notin S} \Cov(v,\E)$, \tpa{} algorithm has $O(\epsilon^{-2} n k \log n)$ time complexity.
	%\end{Lemma}
\end{proof}

\vspace{0.03in}

\textbf{Prefix Solution Property.} An interesting characteristic of our \tpa{} algorithms is that any prefix of the returned solution is also $(1-1/e)$-optimal for the same size. That means the first $k' \leq k$ nodes in $S$ is at least $(1-1/e)$ the maximum coverage of $k'$ nodes.
\begin{Lemma}
	Any prefix set $S_{k'}$ of size $k' \leq k$ of $S$ is an $(1-1/e)$-approximate solution of size $k'$ on the generated hyperedges, i.e.
	\begin{align}
		\Cov(S_{k'}, \E_f^{(T)}) \geq \rho_{k'} \emph{\OPT}^{(T)}_{k'}
	\end{align}
	where $\emph{\OPT}^{(T)}_{k'}$ is the maximum coverage of size-$k'$ sets.
\end{Lemma}
This property follows from a corollary of Lemma~\ref{lem:anyz} that the coverage of $S_{k'}$ satisfies $\Cov(S_{k'}, \E_f^{(T)}) \geq (1-(1-1/k)^{k'})z$ and $z$ remains to be an upper-bound for size $k'$.

\subsection{Analysis of threshold $z^*$ for achieving $(1-1/e-\epsilon)$ approximation guarantee}
\textbf{Solution guarantee analysis.} The following theorem states that with a large enough $z$, \tpa{} returns a near-optimal solution for $k$-cover on the hypergraph $\mathcal{H}$. That is for 	\[
z=	z^* = O\left(\epsilon^{-2} k \log n\right),
\] 
the returned solution $S_{z^*}$ will satisfy

	\begin{align}
\Pr\big [\Covw(S_{z^*}) \geq (1-1/e-\epsilon) \eOPTk \big ] \geq 1- \delta.
\end{align}

%Let $z^* = O(\epsilon^{-2} k \log n)$,
%\begin{align}
%	T^{*} = z^*\frac{m}{\OPTk} \leq z^*\frac{n}{k} = O(\epsilon^{-2}n\log n ),
%\end{align}
%where $\OPTk$ is the maximum coverage on the complete $\mathcal{H}$.

For fixed $0<p<1$, e.g., $p=1/n$, and $T^*= O\left(\epsilon^{-2}  k \log n \frac{1}{\OPTk}\right)$, defined later in Eq.~\ref{eq:t*}, we first  prove that the number of hyperedges $T$ is close to $T^*$. 
	\vspace{-0.05in}
	
	\begin{Lemma}
		\label{lem:lem1}
		For $0<p<1$, with $z = z^*$,  $T$, the number of hyperedges acquired by \tpa{} satisfies that 
		\[
		\Pr[ T^* \leq T \leq cT^* ] \geq 1-4/p,
		\]
	 where constant $c = \frac{(1+\epsilon_2)}{(1-\epsilon_2)(1-1/e)}$.
	\end{Lemma}
	\vspace{-0.1in}
	The proof, using concentration inequalities on all subset of size $k$, is presented in the appendix.

	In the following, we prove that within the bounded range of acquired hyperedges, the return solution is at least $(1-1/e-\epsilon)$ of the optimal solution with high probability.
	
	For some fixed constant $\alpha>0$, we divide the  interval $\left[T^*, cT^*\right ]$ into $\lc\log_{1+\alpha}c\rc$ smaller ones $[L_i, U_i]$ with $i$ taking values from $1$ to $ \lc\log_{1+\alpha}c\rc$, i.e. $L_i = T^* (1+\alpha)^{i-1}$ and $U_i = T^* (1+\alpha)^{i}$. Recall that $\alpha$ is a small positive number, i.e. $\alpha = 0.1$. Since $T^* \leq T \leq cT^*$, $T$ falls into at least one of these small intervals. Hence, we will show that when $T$ falls within any interval, the returned seed set is at least $(1-1/e-\epsilon)$ optimal.
	\vspace{-0.05in}
	
	\begin{Lemma}
		\label{lem:lem2}
		For the interval $[T^*, cT^*]$, \tpa{} returns an $(1-1/e-\epsilon)$-approximate solution with probability at least $1-\lc\log_{1+\alpha}c\rc \frac{4}{p}$.
	\end{Lemma}
	\vspace{-0.1in}

We are now ready to prove the main theorem.

\begin{theorem}
	\label{theo:tpa}
	For $0 < \epsilon < 1-1/e, 0 < \delta < 1$, let
	%	\begin{align}
	%		p = 8/\delta, \alpha = 0.1, \kappa = \sqrt{\log p}, \lambda = \sqrt{\log p + \log \Resize{0.5cm}{{n \choose k}}}, \nonumber
	%	\end{align}
	%	\vspace{-0.25in}
	%
	\begin{align}
	%	\label{eq:z*}
	z^* = O\left(\epsilon^{-2} k \log n\right),
	\end{align}
	where the exact form of $z^*$ is given in Eq.~(\ref{eq:z*}) in the proof.
	%	\begin{align}
	%		z = z^* = O(\epsilon^{-2} (k \log n + \log \delta)) \text{ } (Eq.~\ref{eq:z*}), \nonumber
	%	\end{align} 
	\tpa{} with $z = z^*$ and $f(S,d_S,\E_r) \leq f_r(S, d_S, \E_r)$ returns an $(1-1/e-\epsilon)$-approximate solution $S_{z^*}$ with probability at least $(1-\delta)$,
	\begin{align}
	\Pr\big [\Covw(S_{z^*}) \geq (1-1/e-\epsilon) \eOPTk \big ] \geq 1- \delta.
	\end{align}
\end{theorem}
	\begin{proof}
		%[Proof of Theorem~\ref{theo:tpa}]
		From Lemmas~\ref{lem:lem1} and~\ref{lem:lem2}, \tpa{} returns an $(1-1/e-\epsilon)$-approximate solution with probability of at least $1-\frac{4}{p}-\lc\log_{1+\alpha}c\rc \frac{4}{p}$. Thus, in order to provide the quality guarantee with the desirable probability $1-\delta$, we need to set $p$ such that
		\vspace{-0.05in}
		\begin{align}
			& \frac{4}{p}+\lc\log_{1+\alpha}c\rc\frac{4}{p} = \delta \nonumber \\
			\label{eq:p}
			\Rightarrow \text{ } & p = \frac{4(1+\lc\log_{1+\alpha}c\rc)}{\delta}.
		\end{align}
		That completes the proof of Theorem~\ref{theo:tpa}.
	\end{proof}
\vspace{0.03in}

\textbf{Complexity analysis.}
The time and space complexities of \tpa{} depend on the upper-bound function deployed. Assuming that the requirement function $f(S,d_S,\E_r) = d_S + k \cdot \max_{v \notin S} \Cov(v,\E_r)$, which can be computed efficiently by storing the coverage of every node and updating their orders, is employed, we show the time and space complexities of \tpa{} to provide $(1-1/e-\epsilon)$-approximation factor as a direct corollary of Lemma~\ref{lem:dta_complexity}.
\begin{theorem}
	\label{cor:dta_complexity}
	\tpa{} algorithm with $z = z^*$ and the requirement function $f_r(S, d_S, \E_r)$ uses $O(\frac{z^*n}{k}) = O(\epsilon^{-2} n \log n)$ space and runs in $O(z^*n) = O(\epsilon^{-2} n k \log n)$ time.
\end{theorem}
%This function matches the requirement (R2) that $f(S,d_S,\E) \leq d_S + k \cdot \max_{v \notin S} \Cov(v,\E)$. The first requirement (R1) $f(S,d_S,\E) \geq \OPT_k^{(T)}$ holds by: $\Cov^{(T)}(S \cup S_k^{(T)}) \leq d_S + \sum_{v \in S_*^{(T)}}\Cov(v,\E) \leq d_S + k \cdot \max_{v\notin S} \Cov(v,\E)$. Moreover, by the monotonicity of coverage function, $\Cov^{(T)}(S \cup S_k^{(T)}) \geq \Cov^{(T)}(S_k^{(T)}) = \OPT_k^{(T)}$.
% Thus, $d_S + f(S,\E) \geq \OPT_k^{(T)}$.

\textbf{Remark:} Although setting $z = z^*$ in \tpa{} guarantees an $(1-1/e-\epsilon)$-approximation ratio, the fixed value $z^*$ may be too conservative and a much smaller value may suffice to provide a satisfactory solution. We propose an efficient search procedure for large enough $z$ to provide the same approximation guarantee in the next section.

%% file: body/dynamic.tex
%\AddNote{listing-4-end}{listing-7-end}{right}{}
\vspace{-0.08in}
\section{Adaptive Sampling Algorithm}
\label{sec:dynamic}

%A nice feature of \tpa{} algorithm is the simplicity of using a single threshold parameter to control quality of the returned solution. However, the value of $z^*$ that guarantees the overall approximation ratio in \tpa{} may be too conservative and a much smaller value may suffice to provide a satisfactory solution. 
We propose a dynamic threshold algorithm, namely \dta{}, that searches for the smallest sufficient threshold $z$ to provide an $(1-1/e-\epsilon)$-approximate solution. To cope with the dependency arising from multiple guesses of the threshold, we incorporate adaptive sampling technique with \tpa{}.

%TODO: 
%Motivation: Theoretical boundaries such as the number of necessary samples ($T^*$), approximation factor, .etc 
%
%\subsection{Data-dependent Sample Complexity}
%\begin{itemize}
%\item Examples of same-size data sets with different sample complexity??
%\item Type-1 and Type-2 thresholds definitions.
%\end{itemize}

%\begin{algorithm}
%	\caption{Quality Assessment Algorithm}
%	\label{alg:assess}
%	\KwIn{Candidate solution $S$}
%	\KwOut{Approximation ratio $\rho_S$ of $S$}
%	Acquire $n_1$ random hyperedges to find an empirical upper-bound $\hat \mu^*$ on the relative mean of the optimal solution; \\
%	Acquire $n_2$ random hyperedges to find an empirical estimate $\hat \mu_S$ of $S$ relative mean; \\
%	\textbf{return} $\rho_S = \frac{f_L(n_2, \hat \mu_{S}, \delta_2)}{f_U(n_1, \hat \mu^*,\delta_1)}$;
%\end{algorithm} '\algnewcommand{\algorithmicgoto}{\textbf{go to}}%
\algnewcommand{\Goto}[1]{\algorithmicgoto~\ref{#1}}%
\begin{algorithm} \small
	\caption{\dta{} - Dynamic Threshold Algorithm}
	\label{alg:dynamic}
	\KwIn{A selection budget $k$ and $\epsilon, \delta$}
	\KwOut{An $(1-1/e-\epsilon)$-approximate solution set $S_k$}
	Let $z^*$ defined in Theorem~\ref{theo:tpa} and $i_0, \delta'$ follow Eq.~(\ref{eq:i0}), (\ref{eq:delta'});\\
	Let $\beta = 0.1, S_c = \emptyset$ and $\textsf{UB} = 0$;\\
	%	Let ;\\
	\For{$z \in \Big \{ \frac{z^*}{2^{i_0}}, \frac{z^*}{2^{i_0-1}}, \dots, z^* \Big \}$}{
		$d_c = 0$;\\
		\textbf{In parallel do}: \\
		\SetInd{3.2em}{0.8em}		
		%		\BLOCK{
		\qquad $\bullet$ \tikzmark{listing-4-end}$<T_z, S_z, d_z> = \tpa(k, z)$;  \\
		%		\BLOCK{
		\quad \tikzmark{right}\quad$\bullet$ \For{each hyperedge $E_j$ in \emph\tpa(k, z)}{
			\SetInd{0.5em}{0.6em}
			\If{$S_c \cap E_j \neq \emptyset$}{
				$d_c = d_c + 1$;
			}
			$N_j = \#\text{received hyperedges till } E_j$;\\
			\If{$S_c \neq \emptyset$ and $N_j = \lc(1+\beta)^t\rc$}{
				$\textsf{LB} = f_L(N_j, d_c/N_j, \delta', N_j)$;\\
				\If{$\rho_S = \frac{\textsf{LB}}{\textsf{UB}} \geq 1-1/e-\epsilon$}{
					\tikzmark{listing-7-end} \textbf{return} $S_z$;
				}
			}
		}
		$t_{u} = \lc \log_{1+\beta}(T_z) \rc$, $\textsf{UB} = f_U(T_z, z/T_z,\delta',\lc (1+\beta)^{t_{u}} \rc )$;\\
		$S_c = S_z$;\\
	}
	\textbf{return} $S_{z^*}$;\\
\end{algorithm}

\subsection{Adaptive Sampling Algorithm}
To avoid the fixed setting of $z = z^*$ in \tpa{} in order to obtain an $(1-1/e-\epsilon)$-approximate solution, our Dynamic Threshold Algorithm (\dta{}) makes multiple guesses of the threshold $z$ and runs \tpa{} algorithm to find a candidate solution $S_z$. Then, $S_z$ is verified whether it meets the quality requirement of $(1-1/e-\epsilon)$-optimality by an empirical quality assessment (Lines~11-14). The details are given in Algorithm~\ref{alg:dynamic}.

%which uses the \tpa{} as subroutine and guarantees to find an $(1-1/e-\epsilon)$-approximate solution w.h.p. 
%\dta{} assumes a procedure to assess the approximation guarantee of the any solution set $S$ w.h.p ( Algorithm~\ref{alg:assess}, Subsection~\ref{subsec:error_bound}).
\dta{} runs on a geometric grid of guessed thresholds $z$ (Line~3).
%The motivation behind is that the threshold $z^*$ to provide the $(1-1/e-\epsilon)$-approximate solution w.h.p. may be too large and not efficient. Smaller thresholds may find a satisfactory solution and make the algorithm run faster, however, \tpa{} fails to provide the approximation guarantee. 
%and uses a dedicated quality assessment procedure (presented in Subsection~\ref{subsec:error_bound}) to find a tight empirical quality bound for the solution $S_z$ returned by \tpa{}. 
%The largest threshold taken into consideration is $z^*$.
%An empirical quality bound $\rho_S$ is computed (Line~8) through calculating a lower-bound on actual coverage of $S_z$ by $f_L(N_2, \deg(S_z, \E)/L_2, \delta')$ (Eq. ) and an upper-bound on optimal coverage by $f_U(N_1, z/L_1,\delta')$.
% These two functions are specified in Lemma~\ref{lem:lu_bounds}.
%\begin{align}
%%	\label{eq:lower_bound}
%	f_L(N,\hat \mu,\delta') = \Resize{6.5cm}{ \min \left\{ \frac{3a - c}{3 + c}, \frac{3a + 2c-ac-\sqrt{c^2(a+2)^2+18ac(1-a)}}{4c+3} \right \}}, \nonumber
%\end{align}
%\begin{align}
%%	\label{eq:upper_bound}
%	f_u(N,\hat \mu,\delta') & = \Resize{6.5cm}{ \max \left\{ \frac{3a + c}{3 - c}, \frac{3a + 4c+ac+\sqrt{c^2(4-a)^2+18ac(1-a)}}{8c+3} \right \}}. \nonumber
%\end{align}
%where $a = \hat \mu$ and $c=\frac{\log(1/\delta')}{N}$.
For each $z$, we run in parallel: 1) \tpa{} algorithm with threshold $z$ to find a new candidate solution $S_z$ that will be checked in the next iteration (Line~16) and 2) the quality assessment of the current candidate solution $S_c$. Note that for the first round $i_0$, the assessment is disabled since there is no candidate solution $S_c = \emptyset$.

\textbf{Quality assessment for $S_c$.} The quality of $S_c$, denoted by $\rho_S$, is evaluated as
\begin{align}
	\rho_S = \frac{\textsf{LB}}{\textsf{UB}},
\end{align}
where
\begin{itemize}
	\item $\textsf{UB} = f_U \big(T_z, z/T_z, \delta', \lc (1+\beta)^{t_u} \rc \big)$ provides an upper-bound of $\OPTk{}$ with probability $1-\delta'$. This bound is applied on the closest point larger than $T_z$ on a geometric grid of the number of hyperedges. It is computed once after every run of \tpa{} with a threshold $z$ (Line~14).
%	the optimal coverage with probability at least $1-\delta'$. The value $z/T_z$ is used as a upper-bound of the coverage of the optimal solution $S^*$ since it is unknown. 
	\item $\textsf{LB} = f_L(N_j, d_c/N_j, \delta', N_j)$ provides a lower-bound on the quality of candidate solution $S_c$ with probability $1-\delta'$. This function is called on a geometric grid on hyperedges observed by \tpa{} (Lines~10-13).
\end{itemize}
The detailed formulations are provided in Eq.~(\ref{eq:upper_bound}) and~(\ref{eq:lower_bound}) in our appendix and satisfy the following properties.
\begin{Lemma}
	\label{lem:lu_bounds}
	The lower-bound $\textsf{LB} = f_L(N_j, d_c/N_j, \delta', N_j)$ and upper-bound $\textsf{UB} = f_U \big(T_z, z/T_z, \delta', \lc (1+\beta)^t \rc \big)$ satisfy
	\begin{align}
		\Pr\left[ \textsf{LB} \leq \Covw(S_z) \right] \geq 1-\delta',
	\end{align}
	\vspace{-0.15in}
	\begin{align}
		\Pr\left[ \textsf{UB} \geq \emph{\OPTk} \right] \geq 1-\delta'.
	\end{align}
\end{Lemma}
\dta{} compares $\rho_S$ with the desired approximation guarantee $(1-1/e-\epsilon)$ and stops when $\rho_S \geq (1-1/e-\epsilon)$ (Line~12).

\textbf{Internal constant settings.} The two constants $i_0$ and $\delta'$ decides the first threshold $\frac{z^*}{2^{i_0}}$ and the success probability of \dta{}. We set $i_0$ to be the smallest integer such that $z^*/2^{i_0} \geq (2+2/3\epsilon)\log(1/\delta)/\epsilon^2$ which is the threshold on the optimal coverage \cite{Dagum00} if we want to estimate it with high precision:
\begin{align}
	\label{eq:i0}
	i_0 = \lc\log_2\frac{z^*\epsilon^2}{(2+2/3\epsilon)\log(1/\delta)}\rc,
\end{align}
\vspace{-0.3in}

\begin{align}
	\label{eq:delta'}
	\delta'=\frac{\delta}{2\log_2(z^*) \log_{1+\beta}(cT^*)},
\end{align}
where $T^*$ and $c$ are defined in Theorem~\ref{theo:tpa} with $\OPTk$ set to the smallest possible value $k$ (to make it computable).

%The function $f_U(L_1, z/L_1,\delta')$ is applied on hyperedges observed by \tpa{} (Lines~4, 8) while new hyperedges are acquired to calculate $f_L(L_2, \deg(S_z, \E)/L_2, \delta_2)$ (Lines~6, 8). Note that $z/L_1$ provides a maximum value of $\deg(S^*, \E_f)/L_1$. The number $L_2$ of hyperedges observed for calculating $f_L(L_2, \deg(S_z, \E)/L_2, \delta_2)$ is taken from a geometric grid of step size $1+\beta$ (Line~5) up to $cT^*$ which is the maximum number of hyperedges observed by \tpa{} with high probability when $z = z^*$ (Lemma~\ref{lem:lem1} in the appendix) while the total size of all the hyperedges is at most $(2 \cdot z \cdot n)$. The maximum number of hyperedges $cT^*$ and their total size $(2\cdot z \cdot n)$ determine the maximum time and space spent for verifying the solution.

\vspace{-0.08in}
\subsection{The Analysis}
\textbf{Adaptive Sampling.}
The \dta{} algorithm generates hyperedges adaptively which means that the test (Line~4 in Algorithm~\ref{alg:dynamic}) performed on the hyperedges generated so far determines whether it terminates right there or more hyperedges are necessary. This introduces dependency between the stopping time at which the algorithm terminates and the correctness of achieving an $(1-1/e-\epsilon)$-approximate solution. This stopping time is also a random variable depending on the obtained hyperedges. Thus, we cannot apply a typical concentration bound.

To cope with the random stopping time, we adopt concentration inequalities on weakly-dependent random variables from martingale theory. These inequalities transform the bounds holding at a deterministic time $t$ to bounds that hold at every point of time up to $t$ (see Lemma~\ref{lem:bounds}).

%\textbf{Coupling \tpa{} and quality assessment.}
%In Algorithm~\ref{alg:dynamic}, the lower-bound of $S_z$ is computed separately by observing more hyperedges, however, this procedure can be embedded in the subsequent call to \tpa{} algorithm with threshold $2\cdot z$ that finds the next candidate solution. Then, the algorithm computes $f_L(L_2, \deg(S_z, \E)/L_2, \delta_2)$ and $\rho_S$ at geometric points on the number of hyperedges observed and skips the check on the total size of hyperedges since it is guaranteed by \tpa{} with threshold $2 \cdot z$. The coupling saves half of the necessary samples in expectation.

Based on Lemma~\ref{lem:lu_bounds}, we prove quality guarantee of \dta{}.
\begin{theorem}
	\label{lem:dta_quality}
	Given $0 < \epsilon < 1-1/e, 0 < \delta < 1$, \dta{} returns an $(1-1/e-\epsilon)$-approximate solution with a probability of at least $1-7\delta/3$.
\end{theorem}

The proof are presented in our appendix. Thus, to achieve the desired solution guarantee with a specific probability $1-\tilde\delta$, we set $\delta = 3\tilde\delta/7$ (used in our experiments). The space and time complexities of \dta{} are bounded as follows.

%The approximation guarantee of $(1-1/e-\epsilon)$ of \dta{} follows directly from the stopping condition in Line~4 of the Algorithm~\ref{alg:dynamic}.
%\vspace{-0.05in}
\begin{theorem}
	\label{thm:dta_complexity}
	\dta{} uses only $O(\epsilon^{-2}n \log n)$ space and runs in time $O(\epsilon^{-2}kn \log n)$.
\end{theorem}

\begin{proof}
Since the space complexity of \tpa{} with threshold $z$ is $O(\frac{zn}{k})$ and \dta{} is dominated by the last call to \tpa{} with the maximum threshold $z^*$ in term of memory usage, \dta{} uses $O(\frac{z^* n}{k}) = O(\epsilon^{-2} n \log n) = \tilde{O}(\epsilon^2 n)$ space. Note that the quality assessment does not incur extra space for storing hyperedges.

As shown by Lemma~\ref{lem:dta_complexity}, assuming requirement upper-bound function, each run of \tpa{} with threshold $z$ uses $O(z \cdot n)$ time. Note that the quality assessment uses the same amount of time $O(z \cdot n)$. Thus, the time complexity of \dta{} is $O((\frac{z^*}{2^{i_0}} + \frac{z^*}{2^{i_0-1}} + \dots + z^*) \cdot n) = O(z^* n) = O(\epsilon^{-2} n k \log n)$.
%	Thus, \dta{} has the same space and time complexities as \tpa{}.
\end{proof}

\subsection[Section title sans citation]{Comparison to Recent Adaptive Schemes \cite{Nguyen163,Tang18}}
Adaptive sampling was used in \SSA{} and \DPIMA{} \cite{Nguyen163} \footnote{The original version of the papers \cite{Nguyen163} contains flaws \cite{Huang17} that have been addressed in the extended version \cite{dssa_fix}}. The main idea is an \emph{out-of-sample validation scheme} that uses two separate pools of samples: one for finding solution and one for validating the found solution. The algorithm returns the solution if passing the validation, otherwise, it doubles the number of hyperedges and starts over. 
	
\DPIMA{} also \emph{construct an upper bound on the optimal solution} using the candidate solutions and the fact that the candidate solution is $1-1/e$ optimal on the first pool of samples. The \emph{lower-bound on the candidate solution} is found via a concentration inequality using the second pool of samples \cite{dssa_fix}.

A recent work in \cite{Tang18} also follows out-of-sample-validation approach and improves the upper-bound on the optimal solution using the online bound from Leskovec \cite{Leskovec07} and improve the lower-bound on the candidate solution via a slightly stronger concentration inequality.

%\DPIMA{} is better than \SSA{} in terms of hyperedge utilization since \DPIMA{} reuses the hyperedges for estimation in the next round to find a new candidate.

Although our proposed \dta{} algorithm shares the idea of adaptive sampling and empirical quality bound, there are critical difference in our approach
\begin{itemize}
	\item \textbf{Space efficiency.} \dta{} uses a reduced sketch to provide the same approximation gurantee in only $\tilde{O}(n)$ space, a factor $k$ reduction comparing to \DPIMA{}, \SSA{}, and other methods that  store all samples before invoking the greedy algorithm for max-coverage.
	\item \textbf{Stronger assessment bound.} \dta{} uses a tighter quality bound (Lemma~\ref{lem:lu_bounds}) and assessment methods that further reduces the number of samples.  In this work, we provide new data structure to efficiently update the online bound from \cite{Leskovec07,Tang18}, termed top-$k$ function. We also derive other bounds including requirement functions and dual-fitting bounds. 
\end{itemize}
\vspace{-0.12in}

\vspace{-0.08in}
\section{Extensions}
\label{sec:extensions}
We discuss different upper-bound functions for \tpa{} algorithm and extend it for the budgeted and weighted versions of the $k$-cover problem and for a streaming data model.

\subsection{Other Upper-bound Functions}
\label{subsec:ub_fun}
We present 3 other upper-bound functions that satisfy the requirements for \tpa{} and analyze their complexities.
%the summaries in Table~\ref{tab:ubound_fun}.% and Figure~\ref{fig:offline_t}

%\renewcommand{\arraystretch}{0.7}
%\setlength\tabcolsep{1.5pt}
%\begin{table}[!h]
%	%	\vspace{-0.1in}
%	\caption{Updating costs of Upper-bound functions for $z=z^*$. $s_f=O(kn\log n)$ is the total size of hyperedges (Avg. Tightness refers to experimental average from Table~\ref{tab:ubounds2}).
%		%($\mathcal{C}_E = \sum_{E_j \in \mathcal{E}} |E_j|$).
%	}
%	\vspace{-0.1in}
%	\label{tab:ubound_fun}
%	\centering
%	\begin{tabular}{lL{1.3cm}ccr}
%		\toprule
%		\textbf{Bound} & \textbf{Space} & \textbf{Time} & \textbf{Total Time} & \textbf{Avg. tightness}\\
%		\midrule
%		\textsf{Req.} & $O(n \log n)$ & $O(1)$ & $O(s_f)$ & 75.4\% \\
%		\textsf{Top-$k$} & $O(n \log n)$  & $O(1)$ & $O(s_f)$&77.9\% \\
%		\textsf{DF-2D} & $O(n \log n)$ & $O(1)$ (expected) & $O(s_f)$&89.2\% \\
%		\textsf{DF} & $O(s_f)$ & $O(s_f)$ & $O(s_f^2)$& 98.0\% \\
%		\textsf{LP} & $O(s_f)$ & $O(s_f^4)$ & $O(s_f^5)$& 99.7\%\\
%		\textsf{ILP} & $O(s_f)$ & $2^{poly(n)}$ & $2^{poly(n)}$&100.0\%\\
%		\bottomrule
%	\end{tabular}
%\end{table}
%\begin{figure}[!ht]
%	\vspace{-0.1in}
%	\includegraphics[width=0.6\linewidth]{figures/offline_t}
%	\vspace{-0.15in}
%	\caption{Offline computing time vs. Tightness.}
%	\label{fig:offline_t}
%	\vspace{-0.1in}
%\end{figure}
\vspace{0.03in}

\textbf{Top-$k$ function $f(S,d_S,\E_r) = d_S + \sum_{v \in top-k} \Cov(v,\E_r)$.}
The $\topk$ is the set of $k$ nodes with largest coverages over the reduced sketch $\mathcal{E}_r$.  Similarly to the argument for the requirement function, this \topk{} upper bound satisfies $\OPT \leq d_S + \sum_{v \in \topk{}} \Cov(v,\E_r) \leq d_S + k \cdot \max_{v\notin S} \Cov(v,\E_r)$. This upper-bound function is a special case of the online bound in \cite{Leskovec07} that is applicable for any monotone submodular function.
%Thus, $f(S,d_S,\E) = d_S + \sum_{v \in \topk{}} \Cov(v,\E)$ is a valid upper-bound function for the \tpa{} algorithm.

Unlike the requirement function, that can be updated within linear time during insertion/deletion of hyperedges, a naive approach for updating the Top-$k$ upper-bound can incur an $O(k)$ cost per insertion/deletion of hyperedges, making the \tpa{} inefficient.

%\begin{Lemma}
%	For $z = z^*$ and $f(S,d_S,\E) = d_S + \sum_{v \in \topk{}} \Cov(v,\E)$, \tpa{} algorithm uses $O(\epsilon^{-2} n \log n)$ space and $O(\epsilon^{-2} n k \log n)$ time.
%\end{Lemma}
%TODO: Summarize and move the details to appendix
\underline{Fast Top-$k$ function update using stepwise-heap:} 
To compute $\sum_{v \in \topk{}} \Cov(v,\E_r)$ quickly and the node with maximum coverage (Lines~7), we use a data structure, called \textit{Stepwise-Heap}, that maintains an increasing order of nodes according the their hyperdegrees. When adding/removing a hyperedge $E_j$, the stepwise-heap can be updated in a linear times in terms of $|E_j|$. This data structure is \emph{essential to keep the time complexity for the \tpa{} algorithm low}, i.e., $O(\epsilon^{-2} n k \log n)$.

Specifically, our stepwise-heap stores nodes in $V$ using an array. The nodes are kept in a \emph{non-decreasing order} of their (hyper) degrees. All nodes with the same (hyper) degree $i$ belong to the same ``bucket'' $i$ and the starting position of the bucket $i$ is tracked by a variable $b_i$, for $i=0..z$. Initially, all the nodes are in bucket~$0$.

\emph{Adding/removing a hyperedge $E_j$}: When a hyperedge is added (Line~6), the degree of a node $u$ node in the hyperedge  will increase  from $d(u)$ to $d(u)+1$ and will be moved from the bucket $d(u)$ to $d(u)+1$.  In fact, we only need to swap $u$ with the node at the end of bucket $d(u)$ and decrease the value of $b(d(u)+1)$ by one.  This takes $O(1)$ time for each node in $E_j$ and $O(|E_j|)$ in total. We can remove an hyperedge in a similar way, swapping each node $u \in E_j$ with the node at the beginnning of bucket $d(u)$ and increase $b(d(u))$ by one.

\emph{Updating the value of Top-$k$ function}:  Initially,  top-$k$ sum is zero. We keep track the bucket and the position of the top-$k$ during addition/removal of hyperedges. During the update (either adding or removing),  each node $u \in E_j$ will change the value of the sum of top-$k$ by at most one, and can be kept tracked in $O(1)$. Thus, the sum of top-$k$ can be kept tracked in $O(|E_j|)$ during the adding/removing of $E_j$.
\vspace{0.03in}

\textbf{Dual Fitting Upper-bound Functions.} We can find the exact bound that matches the optimal value using Integer Linear Programming (\textsf{ILP}), unfortunately the time complexity would be exponential and cannot solve even small problems. We devise a dual fitting bound that bases on the dual form of the Linear Programming relaxation (LP) of the problem. The LP relaxation and dual forms are as follows.

\noindent \begin{minipage}{0.48\linewidth} \small
	\vspace{0.12in}
	\textit{Primal form (LP)}
	\vspace{-0.01in}
	\begin{align}
	& \max \sum_{i = 1}^{m} y_i \text{ \ \ \ s.t. } \nonumber\\
	& y_i - \sum_{v_j \in E_i} x_j \leq 0 \text{ \ } (\forall i = 1..m) \nonumber \\
	& \sum_{j = 1}^{n} x_j \leq k, \text{ } x_j \geq 0 \text{ } (\forall j=1..n) \nonumber\\
	& 0 \leq y_i \leq 1, \text{ } (\forall i = 1..m) \nonumber
%	& 0 \leq x_j \text{ } (\forall j=1..n) \nonumber
	\end{align}
	\vspace{-0.1in}
\end{minipage}
\begin{minipage}{0.5\linewidth} \small
	\vspace{-0.1in}
	\textit{Dual form (dual fitting - \textsf{DF})}
	\vspace{-0.01in}
	\begin{align}
	& \min \Big (kt + \sum_{i = 1}^{m}z_i \Big ) \text{ \ \ \ s.t. } \nonumber \\
	& t - \sum_{E_i \ni v_j} f_i \geq 0 \text{ \ } (\forall j = 1..n) \nonumber\\
	& z_i + f_i \geq 1 \text{ \ } (\forall i = 1..m) \nonumber \\
	& t, z_i, f_i \geq 0 \text{ \ } (\forall i = 1..m) \nonumber
	\end{align}
	\vspace{-0.1in}
\end{minipage}

Any feasible solution of the dual program provides an upper-bound due to the weak duality theorem and the optimal bound of this type is obtained by an exact LP solver. We call this bound \textsf{LP} which has high time complexity \cite{Borgwardt87}.

We develop a fast algorithm, namely \textsf{DF} (dual fitting), that does not solve the dual program to optimality but delivers good practical performance. The algorithm starts with the setting of variables: for those hyperedges $E_i$ covered by $S$, set $z_i = 1, f_i = 0$; for the rest of hyperedges, $z_i = 0, f_i = 1$. It consists in two operations: \textit{Increase-$t$} and \textit{Decrease-$t$} that adjust the value of $t$ and $f_i, z_i$ to a feasible solution while trying to reduce the value of objective function $\left (kt + \sum_{i = 1}^{m}z_i \right )$.

\textbf{2-Dimensional Relax Dual Fitting.} The dual form can be simplified based on an observation that for the optimal solution, all the constraints $z_i + f_i \geq 1, \forall i = 1..m$ must be tight. That is because if there is an untight constraint $z_i + f_i > 1$, we can decrease the value of $z_i$ to make it tight resulting in better objective value. Thus, we can eliminate the variables $f_i$ completely and replace $f_i$ by $1-z_i$ with a range constraint that $0 \leq z_i \leq 1$.
%\begin{align}
%	& \text{ } \min \left ( kt + \sum_{i = 1}^{m}z_i \right ) \\
%	\text{s.t. } & t - \sum_{E_i \ni v_j} (1-z_i) \geq 0 \text{ \ } (\forall j = 1..n)\\
%	& t \geq 0, 0 \leq z_i \leq 1 \text{ \ } (\forall i = 1..m) \nonumber
%\end{align}

%The resulted linear program contains fewer variables but it is still too large to be solved efficiently. 
We relax the resulted dual program to reduce the size all the way down to a 2-dimensional problem
% and more importantly, open up an extremely efficient algorithm to solve it consecutively when hyperedges are added and deleted from hypergraph 
as follows: \textit{Assign value $\alpha \in [0,1]$ to variables $z_i$ that their corresponding hyperedges $E_i$ are covered by $S$ and assign 0 for the others}. Hence, we obtain the following 2-dimensional relaxed dual program with two variables, $\alpha$ and $t$.

\textit{2-Dimensional dual fitting (\textsf{DF-2D})}
\begin{align}
	& \text{ } \min \left( kt + \alpha d_S \right )\\
	\text{s.t. } & t - (D_i - \alpha S_i) \geq 0 \text{ \ } (\forall i = 1..n)\\
	& t \geq 0, 0 \leq \alpha \leq 1 \text{ \ } \nonumber
\end{align}
%\begin{align}
%	& \text{ } \min_{\alpha \in [0,1]} \max_{v \in V} \left \{ k D_v - \alpha (k S_v - d_S) \right \}
%\end{align}

Note that by setting $\alpha = 1$, the resulted objective value comes back to the requirement function. Thus, the optimal $\alpha$ at least gives a tighter upper-bound than the first requirement function.

The 2D dual program can be solved efficiently in linear time $O(n)$ and constant time to update if a new constraint is added into the program \cite{Megiddo84}. 
%However, in practice repeatedly solving it after a hyperedge is added or deleted is still expensive. 
Similarly to the two other upper-bound functions, \textsf{DF-2D} requires only $O(\epsilon^{-2} n \log n)$ space and constant time to update in expectation \cite{Megiddo84}.

\subsection{Budgeted $k$-cover}

 \begin{table*}%\scriptsize
 	\centering
 	\caption{Comparison of reduced sketch (\rs{}) and full sketch (\fs{}) on different problems for $k = 100$ and $\epsilon = 0.05$.}
 	\vspace{-0.1in}
 	\begin{tabular}{l r r r rrr rr rr}
 		\toprule
 		& & & & \multicolumn{3}{c}{Sketch Size} & \multicolumn{2}{c}{Time (s)} &
 		\multicolumn{2}{c}{Quality (\%)}\\
 		\cmidrule(lr){5-7} \cmidrule(lr){8-9} \cmidrule(lr){10-11}
 		\textbf{Problems} &\textbf{Data} & \#nodes & \#edges & \rs{} & \fs{}& Reduction &\rs{} & \fs{} & \rs{} & \fs{}\\
 		\midrule
  		Dominating Set& DBLP & 655K & 2M & \textbf{21M} & 59M& 2.8x &\textbf{3} & \textbf{3} &{22} & {22}\\
 		 & Orkut & 3M & 234M &\textbf{240M} & 18.3G& 76x & \textbf{47} &1048 & 87 &{89}\\
 		 & Twiter & 41.7M & 1.5B & \textbf{53M} & {7.6G}& 144x &\textbf{27} &457 & 73 & {74}\\
 		 \midrule
 		 Influence Maximiation & DBLP & 655K & 2M & \textbf{54M} & 152M & 2.8x&18.7 &{\textbf{16}} & 0.31 & {0.32}\\
 		 & Orkut & 3M & 234M &\textbf{20.8M} &858M & 41x&\textbf{152} &1723 & 69 & 68\\
 		 & Twiter & 41.7M & 1.5B &\textbf{2.7G} &61G& 22x &9614 & \textbf{9573} &24 & 25\\
 		 \midrule
 		 Lanmark Selection& DBLP & 655K & 2M &\textbf{1M} &3.9M & 3.9x&\textbf{78} & 90 &22 & 22\\
 		 & Orkut & 3M & 234M &\textbf{496K} &4.3M & 8.7x &110 & \textbf{108} &62 &67\\
 		 & Twiter & 41.7M & 1.5B &\textbf{6.8M} &58M & 8.5x&\textbf{1216} &1438 &58 &60\\
 % 		 \midrule
 % 		 K-path & DBLP & 655K & 2M &\textbf{2.8M} &7.5M & 2.6&\textbf{1} &\textbf{1} &14 &15\\
 % 		 & Orkut & 3M & 234M & \textbf{32M} &67M& 2.1 &22 &\textbf{21} &5 & 5\\
 % 		 & Twiter & 41.7M & 1.5B &\textbf{1M} &2.2M & 2,2&16 &\textbf{14} &36 &36\\
 		\bottomrule
 	\end{tabular}
 	\label{tab:rtime}
 	\vspace{-0.1in}
 \end{table*}
The budgeted $k$-cover imposes a cost $c_v$ of selecting node $v$ and has a budget $L$ for the total cost of all the selected nodes. That means we can only select a set of nodes whose total cost is at most $L$. Without loss of generality, assume $L \geq c_v, \forall v \in V$, otherwise, we can just ignore those nodes with cost more than $L$. We adapt \tpa{} algorithm for this budgeted version to return an $(1-1/\sqrt{e})$-approximate solution as follows:
\begin{itemize}
	\item Change the second requirement (R2) on the upper-bound function from $f(S,\E) \leq k \cdot \max_{v\notin S} \Cov(v,\E)$ to $f(S,\E) \leq L \cdot \max_{v \notin S \cup Q} \frac{\Cov(v,\E)}{c_v}$ where $Q$ is the set of nodes that have maximum ratio of coverage over cost but are not added to $S$ due to their costs exceed the remaining budget $L - c_S$. To satisfy this requirement, the similar three possible upper-bound functions can be applied: 
	\begin{itemize}
		\item The requirement function which matches the second condition above.
		\item The \textsf{top-k} function when nodes are ranked according to the ratio of coverage to cost and $k$ is the first value such that $\sum_{v \in \topk{}} c_v \geq L$.
		\item A dual bound is derived similarly to the case of unweighted by solving a linear program in 2D.
	\end{itemize}
	\item A node $v$ is added to $S$ when: 1) the condition in Line~5 of Alg.~\ref{alg:tpa} is met; 2) $v$ is the node with maximum ratio of coverage over its cost; and 3) the cost $c_v$ does not exceed the remaining budget $L - d_S$. The algorithm terminates when all the remaining unselected nodes have costs greater than the leftover budget $L - c_S$. During acquiring the hyperedges, it also maintains the node $v_{\max}$ with maximum hyperdegree $d_{\max}$ among all the nodes in $V$ with respect to all the acquired hyperedges. When terminated, the algorithm returns the current solution set $S$ if $d_S \geq d_{\max}$ or $v_{\max}$ otherwise.
	\item The threshold $z^*$ at which \tpa{} returns an $(1-1/\sqrt{e}-\epsilon)$-approximate solution is modified to
	\vspace{-0.1in}
	\begin{equation}
	    	z^* = O(\epsilon^{-2} k_m \log n),
	\end{equation}
	\vspace{-0.03in}
	\noindent where $k_m$ is size of the largest set whose cost is at most $L$. The value of $z^*$ is set by the observation that in the worst case, we need to bound the bad events in (\ref{eq:bound_s}) for all sets of size at most $k_m$. In the uniform case with cost 1 for every node, we consider all sets of size at most $k$.
\end{itemize}
The approximation guarantee and complexity of our adapted \tpa{} are shown as follows.
\begin{theorem}
	\label{lem:budget}
	The \dta{} algorithm adapted for the budgeted $k$-cover problem gives an $(1-e^{-0.5}-\epsilon)$-approximate solution and uses $O(\frac{c_{\max}}{c_{\min}} n \log n)$ memory space where $c_{\max} \text{ and } c_{\min}$ are max and min costs of selecting nodes.
\end{theorem}
\subsubsection{Approximation guarantee in Theorem~\ref{lem:budget}}
%\begin{proof}
Let $r$ be the number of selection iterations at which the stopping condition in Line~5 of Alg.~\ref{alg:tpa} is met until the first set from the optimal solution is considered but not added to $S$ due to causing exceeding the remaining budget $L - c_S$. Let $l$ be the number of sets added to $S$ in the first $r$ iterations. Note that $r \geq 2$ and $l \geq 1$ since every node has cost smaller than $L$. Let $c_i$ be the cost of the $i^{th}$ node added to $S$.

We prove the lemma in two steps:

$\bullet$ \textit{Prove $d^+_i \geq \left [ 1-\prod_{j = 1}^{i} \left( 1-\frac{c_j}{L} \right) \right ] \cdot z$.} From the selection criterion of nodes into $S$, we have that $\forall i = 1, \dots, l+1$,
\vspace{-0.05in}
\begin{align}
& d^-_i + f(S,\E) \geq z \nonumber \\
\Rightarrow \text{ } & d^-_i + L \max_{v \notin S \cup Q} \frac{\Cov(v,\E)}{c_v} \geq z \nonumber \\
\Rightarrow \text{ } & d^-_i + \frac{L}{c_i}(d^+_i - d^-_i) \geq z \nonumber \\
\Rightarrow \text{ } & d^+_i - d^-_i \geq \frac{c_i}{L}(z - d^-_i)
\end{align}
\vspace{-0.15in}

\noindent We prove the statement using induction. First, with $i = 1$, it obviously hold that $d^+_1 \geq \frac{c_1}{L} z$. For any $1 < i \leq l+1$, we have
\vspace{-0.05in}
\begin{align}
d^+_i & = d^-_i + (d^+_i - d^-_i)\geq d^-_i + \frac{c_i}{L}(z - d^-_i) \nonumber \\
&\geq \left( 1-\frac{c_i}{L} \right) d^-_i + \frac{c_i}{L}z \geq \left( 1-\frac{c_i}{L} \right) d^+_{i-1} + \frac{c_i}{L}z \nonumber \\
&\geq \left( 1-\frac{c_i}{L} \right) \left (1-\prod_{j = 1}^{i-1} \left (1-\frac{c_j}{L} \right ) \right )\cdot z + \frac{c_i}{L} \cdot z \nonumber \\
\label{eq:d_i}
&\geq \left ( 1-\prod_{j = 1}^{i} \left( 1-\frac{c_j}{L} \right) \right ) \cdot z
\end{align}
\vspace{-0.05in}

$\bullet$ \textit{Prove $\Cov^{(T)}(S) \geq (1-1/\sqrt{e}) \OPT_k$.} If there exists a node $v$ covering at least $\frac{z}{2}$, then our modified budgeted \tpa{} algorithm will guarantee to return a set with coverage at least $\frac{z}{2}$ or have approximation ratio of at least $\frac{1}{2}$. Otherwise, let consider the case that none of the nodes have coverage more than $\frac{z}{2}$. Under this consideration, if the returned set $S$ has $c_S \leq \frac{L}{2}$, then all other nodes $V \backslash S$ have costs at least $\frac{L}{2}$ (otherwise, it can be added to $S$). Thus, the optimal solution $S^*_z$ on acquired hyperedges contains at most $1$ node in $V \backslash S$. Additionally, since every node has coverage at most $\frac{z}{2}$, we find that $\Cov^{(T)}(S \cap S^*_z) \geq \frac{z}{2}$ implying $d_S \geq \frac{z}{2}$. On the other hand, if $c_S \geq \frac{L}{2}$ or $L \leq 2 c_S$, from inequality (\ref{eq:d_i}), we have,
\vspace{-0.05in}
\begin{align}
d^+_l & \geq \left ( 1-\prod_{j = 1}^{l} \left( 1-\frac{c_j}{L} \right) \right ) \cdot z \nonumber \\
& \geq \left( 1-\left( 1-\frac{1}{2l} \right)^l \right) \cdot z \geq \left( 1-\frac{1}{\sqrt{e}} \right) \cdot z.
\end{align}
%	That completes the proof.
%\end{proof}

\subsubsection{Space complexity in Theorem~\ref{lem:budget}}
By setting $z$ to the modified value of $z^*$, similarly to proof of Theorem~\ref{theo:tpa}, it provides an $(1-e^{-0.5} - \epsilon)$-approximation ratio and uses $O(\frac{c_{\max}}{c_{\min}} n \log n)$-memory space due to the selection criterion,
\vspace{-0.05in}
\begin{align}
	& d_S + f(S, d_S, \E_r) \leq z^* \nonumber \\
	\Rightarrow \text{ } & \frac{\Cov(v,\E_r) \cdot L}{c_v} \leq O(\epsilon^{-2}k_m \log n) \nonumber \\
	\Rightarrow \text{ } & \Cov(v,\E_r) \leq O\left (\epsilon^{-2}\frac{c_v \cdot k_m \log n}{L} \right) \leq O\left(\epsilon^{-2} \frac{c_{\max}}{c_{\min}}  \log n \right)
\end{align}
\vspace{-0.15in}

The last inequality is based on the observations that $c_v \leq c_{\max}$ and $L \geq k_m \cdot c_{\min}$. Thus, the maximum total space is $O\left( \frac{c_{\max}}{c_{\min}} \epsilon^{-2} n \log n \right)$.
The same results hold for \dta{} algorithm on the updated $z^*$.

%% file: body/new_exp.tex
\vspace{-0.1in}
\section{Experiments}
\label{sec:exps}
\label{subsec:effectinvess}
\begin{figure*}[!ht]
	\label{fig:ds}
% 		\caption{Domination Set experiments across different $k$ on DBLP network with $\epsilon = 0.1$.}
    \subfloat[Sketch size - DLBP]{
		\includegraphics[width=0.3\linewidth]{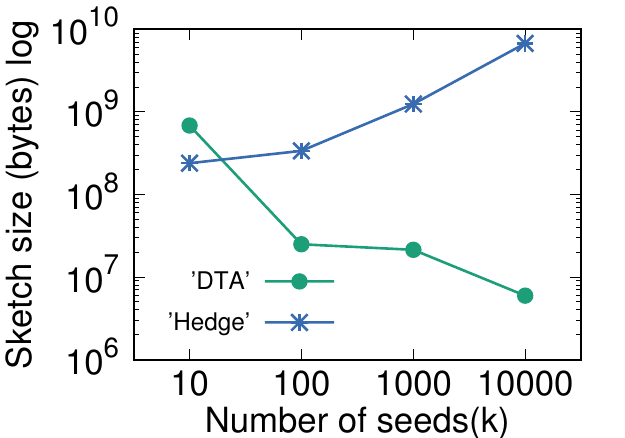}
		%		\label{fig:hep_lt_inf}
	}
	\subfloat[Running time - DLBP]{
		\includegraphics[width=0.3\linewidth]{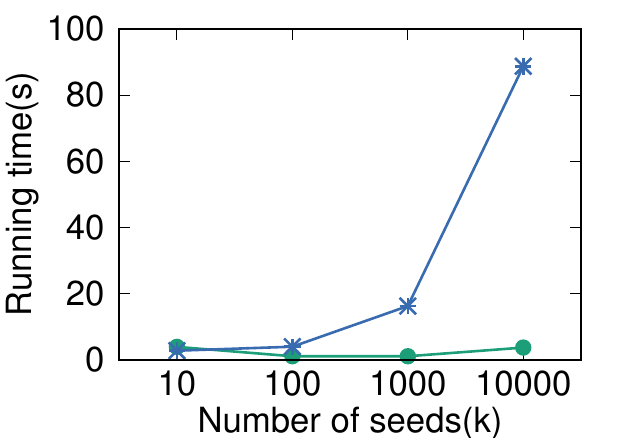}
		%		\label{fig:phy_lt_inf}
	}
	\subfloat[Quality - DLBP ]{
		\includegraphics[width=0.3\linewidth]{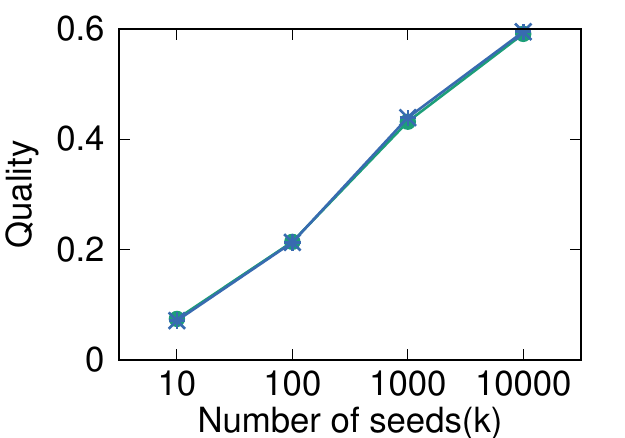}
		%		\label{fig:hep_lt_inf}
	}
	
% 	\vspace{-0.1in}
% 	\label{fig:ds}
% 	\vspace{-0.2in}
% \end{figure*}
% \begin{figure*}[!ht]
\vspace{-0.1in}
	\subfloat[Sketch size - Orkut]{
		\includegraphics[width=0.3\linewidth]{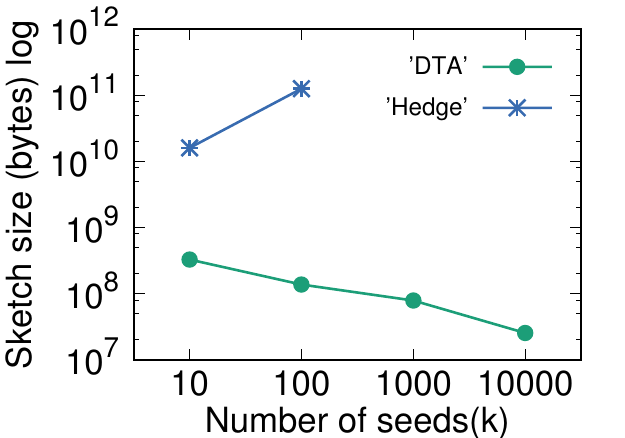}
		%		\label{fig:hep_lt_inf}
	}
	\subfloat[Running time - Orkut]{
		\includegraphics[width=0.3\linewidth]{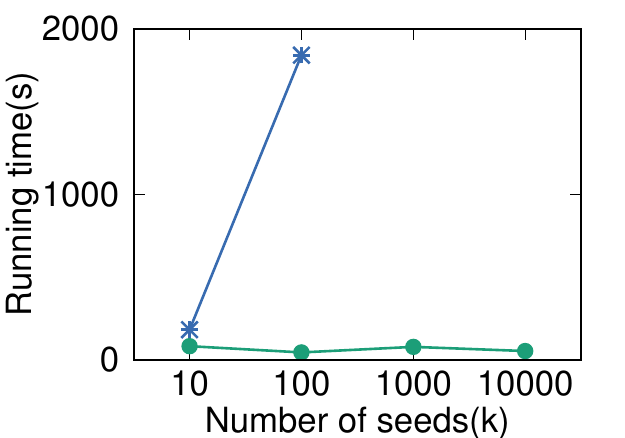}
		%		\label{fig:phy_lt_inf}
	}
	\subfloat[Quality - Orkut]{
		\includegraphics[width=0.3\linewidth]{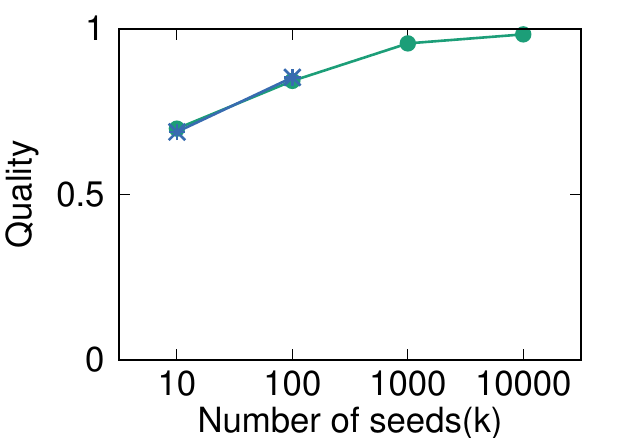}
		%		\label{fig:hep_lt_inf}
	}
% 	\vspace{-0.1in}
% 	\caption{Domination Set experiments across different $k$ on Orkut network with $\epsilon = 0.1$.}
% 	\vspace{-0.2in}
% \end{figure*}
\vspace{-0.1in}
% \begin{figure*}[!ht]
	\subfloat[Sketch size - Twitter]{
		\includegraphics[width=0.3\linewidth]{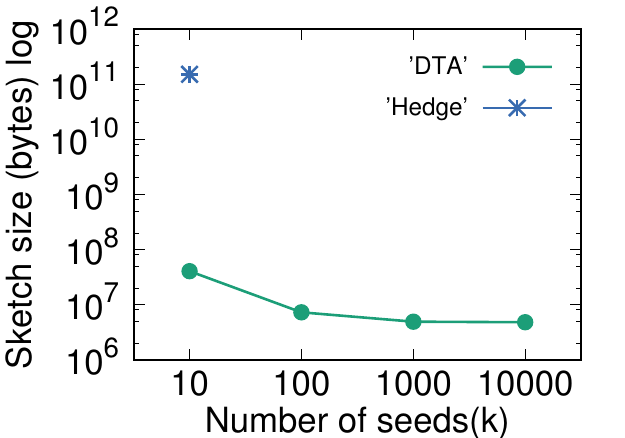}
		%		\label{fig:hep_lt_inf}
	}
	\subfloat[Running time - Twitter ]{
		\includegraphics[width=0.3\linewidth]{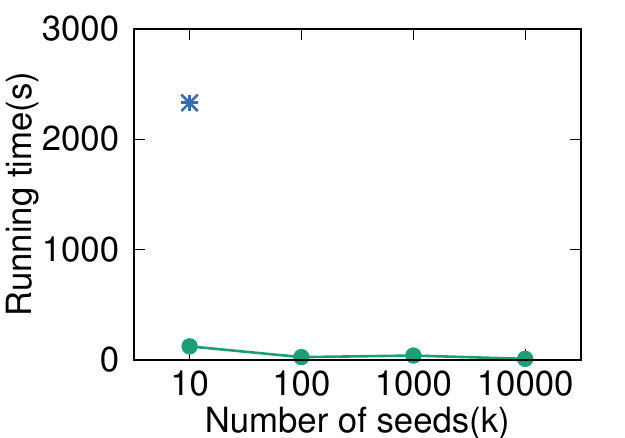}
		%		\label{fig:phy_lt_inf}
	}
	\subfloat[Quality - Twitter]{
		\includegraphics[width=0.3\linewidth]{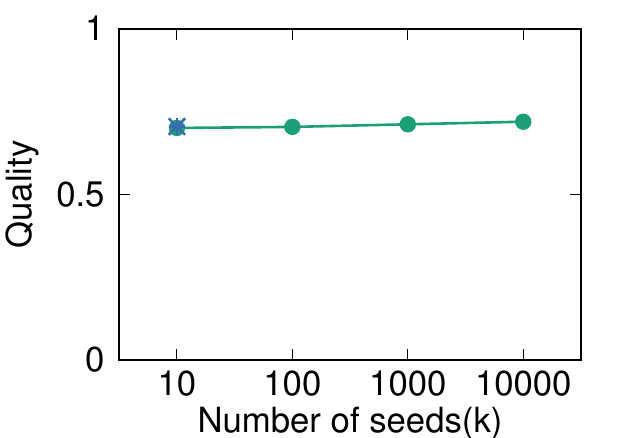}
		%		\label{fig:hep_lt_inf}
	}
	\vspace{-0.1in}
	\caption{Performance for $k$-Dominating Set with varying $k$ and $\epsilon = 0.1$. \dta{} reduces sketch size and time up by factors as large as 1000+ while providing comparable solution quality. }
 	\label{fig:ds}
 	\vspace{-0.15in}
\end{figure*}
We conduct extensive experiments to assess the performance of our proposed approaches on three representative applications: $k$-dominating set, influence maximization, and land mark selection. The definitions of the applications can be found in section \ref{sec:apps}.

First, we assess the effectiveness of our reduced sketch in Subsec. \ref{subsec:effectinvess} following by the comparison of our algorithms with the state-of-the-art approaches for each application in Subsec. 6.2 to 6.4

\noindent \textbf{Metrics.} The comparisons are done using the following three metrics.
\begin{itemize}
    \item \textbf{Running time.} We measure the total running time of each algorithm in seconds.
    
    \item \textbf{Sketch size.} We measure the size of the sketches used by the compared algorithms. The sketch size is measured as \emph{the maximum number of elements} in the sketches (e.g. the total size of the incident lists in the incidence graph) through the run-time, multiplied by the size of the data type to store the elements (8 bytes for integers). 
    
    \item \textbf{Quality.} We measure the quality of a solution $S$ as the (expected) coverage of the solution. For $k$-Dominating set, the coverage is the number of 2-hop neighbors covered by $S$. For Influence maximization \cite{Kempe03}, we measure the quality as the expected fraction of influenced nodes. For Landmark selection, the quality of $S$ is measured as the expected fraction of node pairs with shortest path passing through a node in $S$. The expected values are estimated within a relative error of $1\%$ and a confident $1-1/n$ following the approach in \cite{Nguyen172}.
\end{itemize}

\noindent \textbf{Datasets.}
We evaluate the algorithms for the 3 applications on three popular real-world networks from \cite{snap,Kwak10} with sizes from millions to 1.5 billion edges (see Table \ref{tab:rtime})

\vspace{0.03in}

We obtain the implementations of all algorithms by the authors. \DPIMA{} is the latest version after being fixed \cite{dssa_fix}. In comparison with other algorithms, \dta{} uses the requirement function and later, we study different upper-bound functions and their effects to the performance of \tpa{}/\dta{}.

%\textbf{\dta{} upper-bound functions.} We have pointed out 6 different functions that \dta{} can use, we carry tests on three of them: requirement, $\topk$ and \textsf{DF-2D}, in comparison with other state-of-the-arts and compare the quality of all 6 upper-bound functions at the end. We denote the three versions as \dta-\textsf{req}, \dta{}-{\top} and \dta{}-{2d}.

\vspace{0.03in}

%\vspace{-0.05in}
% \vspace{0.03in}

\noindent \textbf{Experimental Environment.}
All the experiments are run on a cluster of 16 Linux machines, each of which has a 2.30Ghz Intel(R) Xeon(R) CPU E5-2650 v3 40 core processor and 256GB of RAM. We limit the running time of each algorithm to be within 12 hours and memory to 200GB.

%\setlength\tabcolsep{4pt}
%\begin{table}[!htb]
%	%	\vspace{-0.1in}
%	\caption{Compare upper-bound functions}
%	\vspace{-0.1in}
%	\label{tab:ubounds}
%	\centering
%	\begin{tabular}{ l  r  r  r r r }
%		\toprule	
%		\textbf{Data} & \bf \textsf{Req.} & \bf \textsf{Top-$k$} & \bf\textsf{Dual-1} & \bf\textsf{Dual-2} & \bf LP\\
%		\midrule
%		PHY & 7266 & 7002 & 6400 & 4642 & 4504 \\
%		Epinions & 1816 & 1816 & 1600 & 1583 & 1518\\
%		DBLP & 7434 & 7103 & 6400 & 4591 & 4539\\
%		Orkut & 498 & 446 & 400 & 400 & 400 \\
%		Twitter & 299& 299 & 200 & 200 & \textsf{dnf} \\
%		Friendster & 100 & 100 & 100 & 100 & \textsf{dnf}\\
%		\bottomrule
%	\end{tabular}
%	%	\vspace{-0.1in}
%\end{table}

\vspace{-0.12in}
\subsection{Full Sketch vs. Reduced sketch}
 \label{subsec:effectinvess}
We compare our reduced sketch $\E_r$ with the full sketch, in which all hyperedges are stored before invoking the greedy algorithm to solve the max-$k$-coverage.

We first run \dta{} with the basic requirement function and $\epsilon = 0.05, \delta = 1/n, k=100$ and obtain the peak total size of the reduced and full sketch. 
Then on the same set of hyperedges generated by \dta{} without removing any hyperedges, termed full sketch $E_f$, we run the  greedy algorithm \cite{Nemhauser81} on $\E_f$. 

The time, sketch size, and quality for the two sketching methods are shown in Table~\ref{tab:rtime}. Across all three applications, the reduced sketch provide solutions with comparable solution quality and running time. However, the reduced sketch \emph{reduce the space up to several orders of magnitude}. For example, on Twitter, the reduction factors for the three applications are 144x, 22x, and 8.5x, respectively.

\subsection{Results on $k$-Dominating Set}

We compare \dta{} with an adaptation of \hedge{} \cite{Mahmoody16} for the $k$-dominating set problem. The extensions are straightforward by replacing the sampling procedure and keep all the other settings and assumptions to provide the $(1-1/e-\epsilon)$ approximation guarantee.
% \textbf{Results Overview.} Our experimental evaluations are presented in Tables~\ref{tab:rtime}, \ref{tab:mem_relative} and \ref{tab:qual_relative}. Our results evidently show that \dta{} algorithm generally uses 10x less time, several time less memory while obtaining the same level of solution quality with existing state-of-the-art algorithms. 
%Results on the \#samples are presented in our appendix.
\begin{figure*}[!ht]
	\subfloat[Sketch size]{
		\includegraphics[width=0.3\linewidth]{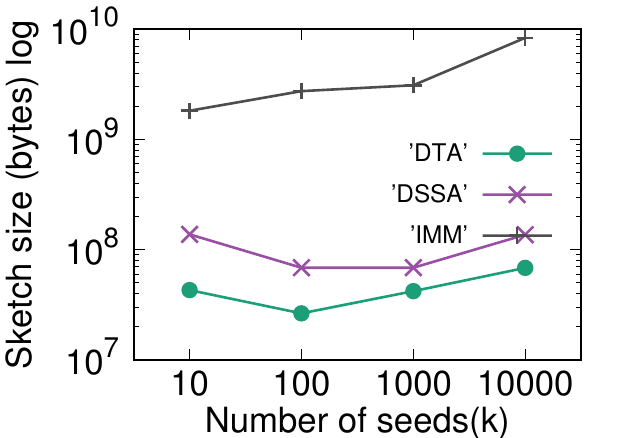}
		%		\label{fig:hep_lt_inf}
	}
	\subfloat[Running time]{
		\includegraphics[width=0.3\linewidth]{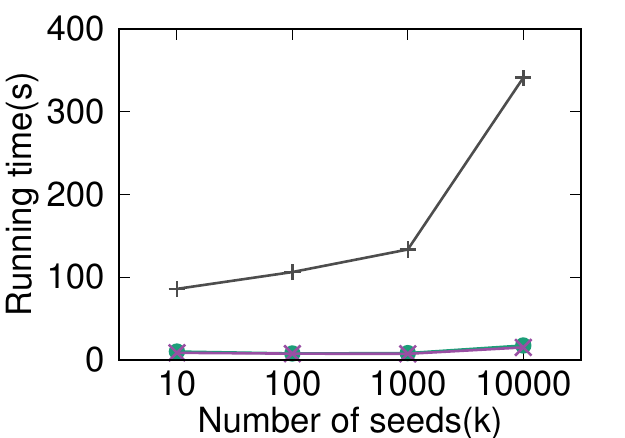}
		%		\label{fig:phy_lt_inf}
	}
	\subfloat[Quality ]{
		\includegraphics[width=0.3\linewidth]{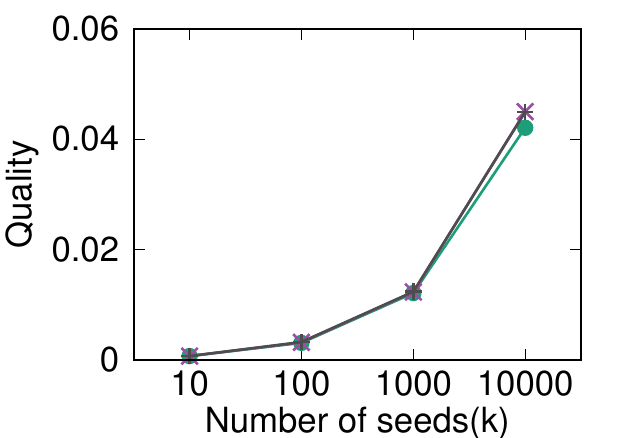}
		%		\label{fig:hep_lt_inf}
	}
	\vspace{-0.1in}
	\caption{Influence Maximization experiments across different $k$ on DBLP network with $\epsilon = 0.1$.}
	\label{fig:im}
	\vspace{-0.2in}
\end{figure*}
\begin{figure*}[!ht]
	\subfloat[Sketch size ]{
		\includegraphics[width=0.3\linewidth]{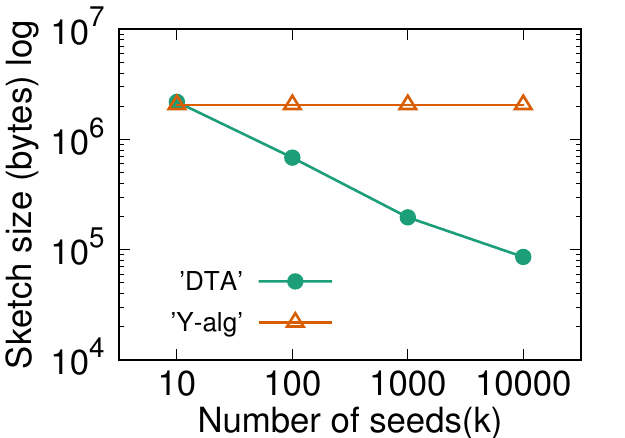}
		%		\label{fig:hep_lt_inf}
	}
	\subfloat[Running time]{
		\includegraphics[width=0.3\linewidth]{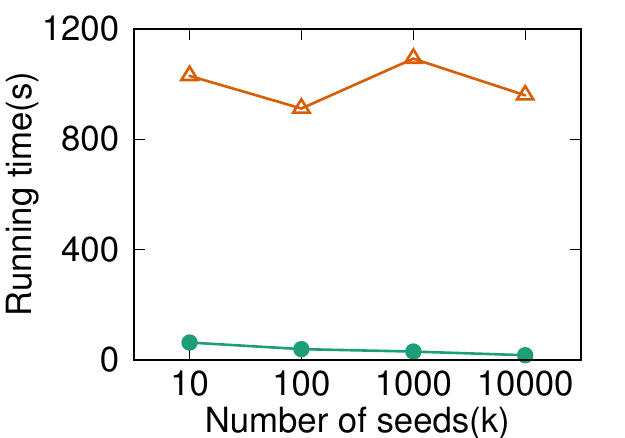}
		%		\label{fig:phy_lt_inf}
	}
	\subfloat[Quality]{
		\includegraphics[width=0.3\linewidth]{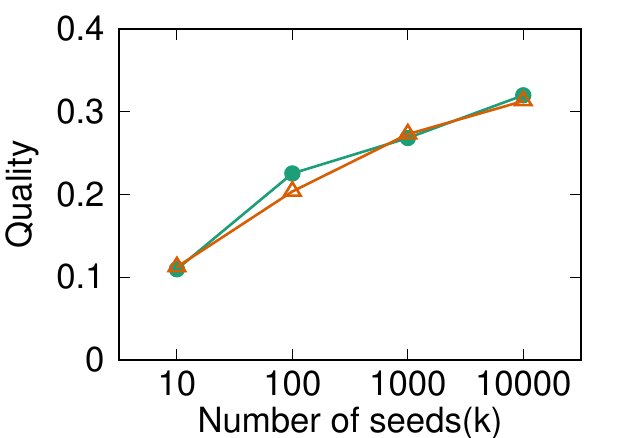}
		%		\label{fig:hep_lt_inf}
	}
	\vspace{-0.1in}
	\caption{Landmark Selection experiments across different $k$ on DBLP network with $\epsilon = 0.1$.}
	\label{fig:cc}
% 	\vspace{-0.2in}
\end{figure*}
The experiment results are shown in Fig.~\ref{fig:ds}.  The proposed \dta{} is several orders of magnitude faster and using several orders less space, while providing comparable solution quality. For example, on Twitter, \dta{} uses 3735x less space and 19x faster for $k=10$. \hedge{} cannot complete within the time/space limit for other values of $k$.

\subsection{Results on Influence Maximization}
% \begin{figure*}[!ht]
% 	\subfloat[Sketch size]{
% 		\includegraphics[width=0.2\linewidth]{figures/im_sketch.pdf}
% 		%		\label{fig:hep_lt_inf}
% 	}
% 	\subfloat[Running time]{
% 		\includegraphics[width=0.2\linewidth]{figures/im_time.pdf}
% 		%		\label{fig:phy_lt_inf}
% 	}
% 	\subfloat[Quality ]{
% 		\includegraphics[width=0.2\linewidth]{figures/im_quality.pdf}
% 		%		\label{fig:hep_lt_inf}
% 	}
% 	\vspace{-0.1in}
% 	\caption{Influence Maximization experiments across different $k$ on DBLP network with $\epsilon = 0.1$.}
% 	\label{fig:im}
% 	\vspace{-0.2in}
% \end{figure*}
% \begin{figure*}[!ht]
% 	\subfloat[Sketch size ]{
% 		\includegraphics[width=0.2\linewidth]{figures/cc_sketch.pdf}
% 		%		\label{fig:hep_lt_inf}
% 	}
% 	\subfloat[Running time]{
% 		\includegraphics[width=0.2\linewidth]{figures/cc_time.pdf}
% 		%		\label{fig:phy_lt_inf}
% 	}
% 	\subfloat[Quality]{
% 		\includegraphics[width=0.2\linewidth]{figures/cc_quality.pdf}
% 		%		\label{fig:hep_lt_inf}
% 	}
% 	\vspace{-0.1in}
% 	\caption{Landmark Selection experiments across different $k$ on DBLP network with $\epsilon = 0.1$.}
% 	\label{fig:cc}
% 	\vspace{-0.2in}
% \end{figure*}
We compare \dta{} with \IMM{} \cite{Tang15} and \DPIMA{} \cite{Nguyen163} \RIS{}-based algorithms that provide $(1-1/e-\epsilon)$ approximation guarantee.

The algorithms are compared on the popular independent cascade (IC) model \cite{Kempe03} and the probabilities on the edges are selected at random from $\{0.1, 0.01, 0.001\}$  following \cite{Cohen14, Chen10, Jung12}. This settings of weights has been shown to be challenging in recent studies \cite{Tang17,Arora17}. 
%, and \textit{Weighted Cascade }(\textsf{WC})  model \cite{Chen09_2,Tang14,Tang15,Nguyen163} that sets the weight of %$(u,v)$ disproportional to in-degree of $v$, i.e. $w(u,v) = 1/d_{in} (v)$.

%Strikingly, the sketch size in our reduced sketch actually decrease as the seed set size increase. The reason is that a larger $k$ results in a smaller selection threshold for \dta{}. Thus, nodes will be selected into the seed set only after a few hyperedges, and more hyperedges are removed as they get covered by the selected nodes.

% \begin{figure*}[!ht]
% 	\subfloat[Sketch size]{
% 		\includegraphics[width=0.25\linewidth]{figures/im_sketch.pdf}
% 		%		\label{fig:hep_lt_inf}
% 	}
% 	\subfloat[Running time]{
% 		\includegraphics[width=0.25\linewidth]{figures/im_time.pdf}
% 		%		\label{fig:phy_lt_inf}
% 	}
% 	\subfloat[Quality ]{
% 		\includegraphics[width=0.25\linewidth]{figures/im_quality.pdf}
% 		%		\label{fig:hep_lt_inf}
% 	}
% 	\vspace{-0.1in}
% 	\caption{Influence Maximization experiments across different $k$ on DBLP network with $\epsilon = 0.1$.}
% 	\label{fig:im}
% 	\vspace{-0.2in}
% \end{figure*}

The results are shown in Fig.~\ref{fig:im}. Unfortunately, none of the two algorithms  \IMM{} \cite{Tang15} and \DPIMA{} \cite{Nguyen16} can complete within the time/space limit on Orkut and Twitter. Thus, only the results for DBLP are available. 

On DBLP, \dta{} use 122x less space than \IMM{} and 2.6x less space than \DPIMA{}. All algorithms produce solutions with similar quality. We note that on bigger network the reduction in space will be much larger. Unfortunately, none of the two algorithms can complete in the larger networks.

% \vspace{-0.08in}
\subsection{Results on Landmark Selection}
% \begin{figure*}[!ht]
% 	\subfloat[Sketch size ]{
% 		\includegraphics[width=0.25\linewidth]{figures/cc_sketch.pdf}
% 		%		\label{fig:hep_lt_inf}
% 	}
% 	\subfloat[Running time]{
% 		\includegraphics[width=0.25\linewidth]{figures/cc_time.pdf}
% 		%		\label{fig:phy_lt_inf}
% 	}
% 	\subfloat[Quality]{
% 		\includegraphics[width=0.25\linewidth]{figures/cc_quality.pdf}
% 		%		\label{fig:hep_lt_inf}
% 	}
% 	\vspace{-0.1in}
% 	\caption{Landmark Selection experiments across different $k$ on DBLP network with $\epsilon = 0.1$.}
% 	\label{fig:cc}
% 	\vspace{-0.2in}
% \end{figure*}
% \vspace{-0.08in}

For landmark selection, we compare \dta{} with \yalg{} \cite{Yoshida14} state-of-the-art algorithm for Landmark selection (aka coverage centrality maximizations). Note that \yalg{} only provide (arguably weaker) addictive theoretical guarantees (see Table~\ref{tab:algorithms}). 

The results are consistent with those in the previous applications. \dta{} uses 24 time less space and runs 59 times faster than \yalg{}, while providing comparable solution quality.

 Overall, our reduced sketch approach reduces the space up to several orders of magnitude while providing comparable running time and solution quality. The proposed algorithm \dta{} with adaptive sampling technique can further reduces the number of samples to provide $1-1/e-\epsilon$ approximation guarantee. Thus, \dta{}, with \emph{optimal approximation guarantee} and \emph{optimal space complexity}, and \emph{practical scalability} provides an ideal solutions for $k$-cover and related applications.
 
%  The results show a significant saving in terms of running time and memory consumption provided by our algorithm compared to the others. This demonstrates our advantages of having small time and space complexities. Comparison between \dta{} and \tpa{} shows that \dta{} uses in orders of magnitude smaller time and memory than \tpa{} proving the effectiveness of our accurate quality assessment deployed in \dta{}.

\vspace{-0.08in}

%% file: body/appendix.tex
%\appendix
\label{sec:appendix}

\vspace{-0.08in}
\section{Concentration inequalities}
We use the following concentration inequality.
\begin{Lemma}
	\label{lem:central_bound}
	Let $X^{(i)}_S, i = 1..N$ be $N$ Bernoulli random variables with the same mean $0 \leq \mu_S \leq 1$ defined in Eq.~(\ref{xsj}). Then for any $x \geq 0$,
	\vspace{-0.05in}
	\begin{align}
	%			\label{eq:bound}
		\Pr \Big [ \max_{1 \leq j \leq N} \big |\sum_{i = 1}^{j}&(X^{(i)}_S - \mu_S) \big | \geq x \Big ] \leq  \nonumber \\
		& 2 \exp \Big (\frac{-x^2}{2(N \mu_S (1-\mu_S) + \frac{1}{3}x)} \Big ). \nonumber
	\end{align}
\end{Lemma}

To prove the above concentration inequality, we use the following general Hoeffding's bounds for supermartingales stated in Theorem~2.1 and Remark~2.1 of \cite{Fan12}:
\begin{Lemma}[\cite{Fan12}]
	\label{lem:hoeffding}
	Assume that $(\xi_i, \mathcal{F}_i)_{i = 1, \dots, N}$ are supermartingale differences satisfying $\xi_i \leq 1$. Let $\Sigma_j = \sum_{i=1}^{j}\xi_i$. Then, for any $x \geq 0$ and $v > 0$, the following bound holds,
	\vspace{-0.05in}
	\begin{align}
		\Pr\left[ \Sigma_j \geq x \text{ and } \langle\Sigma\rangle_j \leq v^2 \text{ for some } j = 1,\dots,N \right] \nonumber \\
		\leq \exp\Big( - \frac{x^2}{2(v^2+\frac{1}{3}x)} \Big),
	\end{align}
	where $\langle\Sigma\rangle_j = \sum_{i=1}^{j}\mathbb{E}(\xi_i^2 | \mathcal{F}_{i-1})$.
\end{Lemma}

By choosing $v^2 = \sum_{i=1}^{N}\mathbb{E}(\xi_i^2 | \mathcal{F}_{i-1})$, we have $\langle\Sigma\rangle_j \leq v^2, \forall j = 1, \dots, N$ and thus,
\vspace{-0.05in}
\begin{align}
\label{eq:hoeffding}
\Pr\Big[ \max_{1 \leq j \leq n} \Sigma_j \geq x \Big] \leq \exp\Big( - \frac{x^2}{2(v^2+\frac{1}{3}x)} \Big).
\end{align}
\vspace{-0.1in}

In our context of $k$-cover, based on random variables $X^{(j)}_S$ defined in Eq.~(\ref{xsj}), we define another set of random variables,
\vspace{-0.05in}
\begin{align}
\xi_j = X^{(j)}_S - \mu_S,
\end{align}
\vspace{-0.15in}

\noindent that satisfy $-\mu_S \leq \xi_j \leq 1 - \mu_S$ and form a sequence of supermartingale differences since
\vspace{-0.05in}
\begin{align}
	\mathbb{E}[\xi_j | \mathcal{F}_{j-1}] = \mathbb{E}\left[X^{(j)}_S - \mu_S\right] = 0,
\end{align}
\vspace{-0.15in}

\noindent where $\{\mathcal{F}_i\}_{i=1,2,\dots}$ is a filtration in which $\mathcal{F}_{i}$ is the $\sigma$-algebra defined on the outcomes of the first $i$ Bernoulli variables $\xi_j, \forall j = 1, \dots, i$.
Thus, the inequality bound (\ref{eq:hoeffding}) simplifies in our case to,
\vspace{-0.1in}
\begin{align}
	\Pr\Big[ \max_{1 \leq j \leq N} \sum_{i = 1}^{j}&(X^{(i)}_S - \mu_S) \geq x \Big] \leq \nonumber \\ &\exp\Big( - \frac{x^2}{2(N \mu_S (1-\mu_S) + \frac{1}{3}x)} \Big). \nonumber
\end{align}

Similarly considering $\xi_j = \mu_S - X^{(j)}_S$, we obtain
\vspace{-0.1in}
\begin{align}
	\Pr\Big[ \max_{1 \leq j \leq N} \sum_{i = 1}^{j}&(X^{(i)}_S - \mu_S) \leq - x \Big] \leq \nonumber \\ & \exp\Big( - \frac{x^2}{2(N \mu_S (1-\mu_S) + \frac{1}{3}x)} \Big). \nonumber
\end{align}

From Lemma~\ref{lem:central_bound}, we derive the number of random variables $X^{(j)}_S$ to provide the concentration bound with a fixed probability $\delta$.
\begin{corollary}
	\label{cor:n_samples}
	Let $X^{(i)}_S, i = 1..N$ be $N$ Bernoulli random variables with mean $0 \leq \mu_S \leq 1$ defined in Eq.~(\ref{xsj}). Let
	\vspace{-0.05in}
	\begin{align}
		N(\epsilon,\delta,\mu) = 2\Big(1+\frac{1}{3}\Big)\frac{1}{\mu \epsilon^2} \ln\Big ( \frac{2}{\delta}\Big).
	\end{align}
	\vspace{-0.1in}
	
	\noindent Then for $\epsilon \geq 0$, $\mu \geq \mu_{S}$ and $N \geq N(\epsilon,\delta,\mu)$,
	\vspace{-0.1in}
	\begin{align}
		\Pr \Big [ \max_{1 \leq j \leq N} \big |\sum_{i = 1}^{j}(X^{(i)}_S - \mu_S) \big | \geq \epsilon \mu N \Big ] \leq \delta.
	\end{align}
\end{corollary}

\vspace{-0.15in}
\section{Omitted Proofs}

First, we introduce some essential definitions used in our proof. Let $0 < \alpha < 1$ be a small constant, e.g. $\alpha = 0.1$ and fixed $0< p < 1$, e.g, $p=1/n$. Define

%	Let $p$ be another constant to be set later to $\frac{4(1+\lc\log_{1+\alpha}c\rc)}{\delta}$ (Eq.~\ref{eq:p}) where $c = \frac{(1+\epsilon_2)}{(1-\epsilon_2)(1-1/e)}$ that \tpa{} obtains an $(1-1/e-\epsilon)$-approximate solution with the desirable probability $1-\delta$.
\vspace{-0.05in}
%	\begin{align}
%		\kappa = \sqrt{\log p}, \lambda = \sqrt{\log p + \log {n \choose k}},
%	\end{align}
%	\vspace{-0.3in}
%	
\begin{align}
\label{eq:epsilon_2_app}
\epsilon_2 & = \frac{\sqrt{\log p + \log \Resize{0.5cm}{{n \choose k}}}}{(1-1/e) \sqrt{\log p} + \sqrt{\log p + \log \Resize{0.5cm}{{n \choose k}}}} \frac{\epsilon}{1+\alpha},\\
\label{eq:c}
c &= \frac{(1+\epsilon_2)}{(1-\epsilon_2)(1-1/e)}, p = \frac{4(1+\lc\log_{1+\alpha}c\rc)}{\delta},
\end{align}
\vspace{-0.3in}

\begin{align}
\label{eq:z*}
z^* = \frac{1+\epsilon_2}{1-1/e}\left (2+\frac{2}{3}\epsilon_2(1-\alpha) \right )\epsilon_2^{-2} \log \left( p{n \choose k} \right ) \nonumber \\ = O\left(\epsilon^{-2} k \log n\right),
\end{align}
\vspace{-0.2in}

\noindent and,
\vspace{-0.05in}
\begin{align}
\label{eq:t*}
T^* = z^* \frac{(1-1/e)(1+\alpha)^2}{(1+\epsilon_2)} \frac{1}{\OPTk} = O\left(\epsilon^{-2}  k \log n \frac{1}{\OPTk}\right).
\end{align}
\vspace{-0.1in}

	\begin{proof}[Proof of Lemma~\ref{lem:lem1}]	
	%	\noindent \textbf{With $z = z^*$, the number of acquired hyperedges is bounded.} We will find both the lower and upper bounds on the number of acquired hyperedges $T$.
	
	%	In particular, we show that the following bounds happen w.h.p.,
	%	\begin{itemize}
	%		\item $T \geq T^*$,
	%		\item $T \leq \frac{(1+\epsilon_2)}{(1-\epsilon_2)(1-1/e)}T^*$.
	%	\end{itemize}
	
	$\bullet$ \textbf{Prove $T \geq T^*$ with probability $1-2/p$.} Consider the case when $T^*$ hyperedges have been acquired. Apply the bound in Corollary~\ref{cor:n_samples} for a set $S$ with $\mu = \mu_{S^*} = \OPTk \geq \mu_{S}$ and note that $T^* \geq N(\epsilon_2,\frac{1}{p {n \choose k}},\OPTk)$,
	%	 with $x = \epsilon_2 \mu_{S^*} T^*$ (recall that $S^*$ is the optimal solution and $\mu_{S^*} = \OPTk$), we get
	%	\vspace{-0.05in}
	%	\begin{align}
	%		\Pr \Big[ \max_{1 \leq j \leq T^*} \big |\sum_{i=1}^{j}(X^{(i)}_S - \mu_S) \big| \geq \epsilon_2 \mu_{S^*} T^* \Big] \leq 2 \exp \Big( -\frac{\epsilon_2^2 \mu_{S^*} T^*}{2(1 + \frac{1}{3}\epsilon_2)} \Big). \nonumber
	%	\end{align}
	%	Plugging $T^* = z^* \frac{(1-1/e)(1+\alpha)^2}{(1+\epsilon_2)} \frac{1}{\OPTk}$ in the above probabilistic inequality and simplifying the right-hand side give
	\vspace{-0.1in}
	\begin{align}
	\label{eq:bound_s}
	\Pr \Big[ \max_{1 \leq j \leq T^*} \big| \sum_{i = 1}^j(X^{(i)}_{S} - \mu_{S}) \big| \geq \epsilon_2 \mu_{S^*} T^* \Big] \leq \frac{2}{p{n \choose k}}.
	\end{align}
	
	The inequality bound (\ref{eq:bound_s}) is essential through the whole proof and reused multiple times later on. From (\ref{eq:bound_s}), since there are at most ${n \choose k}$ sets having exactly $k$ nodes,
	\vspace{-0.1in}
	\begin{align}
	& \Pr \Big[ \max_{1 \leq j \leq T^*, |S| = k} \big| \sum_{i = 1}^{j}(X^{(i)}_{S} - \mu_{S}) \big| \geq \epsilon_2 \mu_{S^*} T^* \Big] \leq \frac{2}{p} \nonumber \\
	%		\Rightarrow & \Pr \left[ \exists j, S,  1 \leq j \leq T^*, |S| = k, \left| \sum_{i = 1}^{j}(X^{(i)}_{S} - \mu_{S}) \right| \geq \epsilon_2 \mu_{S^*} T^* \right] \leq \frac{2}{n} \nonumber \\
	\Rightarrow \text{ } & \Pr \Big[ \max_{1 \leq j \leq T^*, |S| = k} \sum_{i = 1}^{j}X^{(i)}_{S} \geq \mu_{S}\cdot j + \epsilon_2 \mu_{S^*} T^* \Big] \leq \frac{2}{p}. \nonumber
	\end{align}
	\vspace{-0.1in}
	
	Note that $\mu_{S} \leq \mu_{S^*} = \OPTk, j \leq T^*$. Replacing $\mu_{S}$ by $\mu_{S^*}$, plugging $T^*$ and simplifying the above equation give
	\vspace{-0.05in}
	\begin{align}
	%		\text{ } & \Pr \Big[ \exists j, S, 1 \leq j \leq T^*, |S| = k, \sum_{i = 1}^{j}X^{(i)}_{S} \geq (1 + \epsilon_2) \mu_{S^*} T^* \Big] \leq \frac{2}{p} \nonumber \\
	%		\Rightarrow \text{ } &
	\Pr \Big[ \max_{1 \leq j \leq T^*, |S| = k} \sum_{i = 1}^{j}X^{(i)}_{S} \geq (1-\frac{1}{e}) z^* \Big] \leq \frac{2}{p}
	%		\Rightarrow \text{ } & \Pr \left[ \forall j, S, 1 \leq j \leq T^*, |S| = k, \sum_{i = 1}^{j}X^{(i)}_{S} < (1-\frac{1}{e}) z^* \right] > 1 - \frac{2}{n} \nonumber
	\end{align}
	\vspace{-0.1in}
	
	In other words, with a probability of $1-\frac{2}{p}$, there is no set of $k$ nodes that can cover $(1-\frac{1}{e}) z^*$ acquired hyperedges or more. However, since \tpa{} returns a set that covers at least $(1-\frac{1}{e}) z^*$ hyperedges (Lemma~\ref{lem:anyz}) when $z = z^*$. That means that \tpa{} acquires more at least $T^*$ hyperedges with probability $1-2/p$.
	
	$\bullet$ \textbf{Prove $T < \frac{(1+\epsilon_2)}{(1-\epsilon_2)(1-1/e)}T^*$ with probability $1-2/p$.} Similar to (\ref{eq:bound_s}), for the optimal set $S^*$, we have with $T' = \frac{(1+\epsilon_2)}{(1-\epsilon_2)(1-1/e)}T^*$
	\vspace{-0.1in}
	\begin{align}
	& \Pr \Big[ \max_{1 \leq j \leq T'} \big| \sum_{i = 1}^j(X^{(i)}_{S^*} - \mu_{S^*}) \big| \geq \epsilon_2 \mu_{S^*} T' \Big] \leq \frac{2}{p{n \choose k}}. \nonumber
	\end{align}
	Consider the last $j = T'$ and rearrange the terms, we get
	\begin{align}
	%		\Rightarrow \text{ } & \Pr \Big[ \big| \sum_{i = 1}^{T'}(X^{(i)}_{S^*} - \mu_{S^*}) \big| \geq \epsilon_2 \mu_{S^*} T' \Big] \leq \frac{2}{p{n \choose k}} \nonumber \\
	%		\Rightarrow \text{ } & 
	\Pr \Big[\sum_{i = 1}^{T'} X^{(i)}_{S^*} \leq \mu_{S^*} \cdot T' - \epsilon_2 \mu_{S^*} T' \Big] \leq \frac{2}{p{n \choose k}}. \nonumber
	\end{align}
	Placing $T' = \frac{(1+\epsilon_2)}{(1-\epsilon_2)(1-1/e)}T^*$ gives
	\vspace{-0.15in}
	
	\begin{align}
	\text{ } & \Pr \Big[\sum_{i = 1}^{T'} X^{(i)}_{S^*} \leq (1 - \epsilon_2) \mu_{S^*} \frac{(1+\epsilon_2)}{(1-\epsilon_2) (1-1/e)} T^* \Big] \leq \frac{2}{p{n \choose k}} \nonumber \\
	\label{eq:u_bound}
	\Rightarrow \text{ } & \Pr \Big[\sum_{i = 1}^{T'} X^{(i)}_{S^*} \leq z^* \Big] \leq \frac{2}{p{n \choose k}} \leq \frac{2}{p}
	\end{align}
	\vspace{-0.12in}
	
	Thus, with a probability of at least $1-\frac{2}{p}, \sum_{i = 1}^{T'} X^{(i)}_{S^*} > z^*$ happens. However, \tpa{} guarantees that the coverage of any set on the acquired hyperedges is at most $z^*$. That means $T < \frac{(1+\epsilon_2)}{(1-\epsilon_2)(1-1/e)}T^*$ with probability at least $1-2/p$.
\end{proof}
\vspace{-0.05in}

	\begin{proof}[Proof of Lemma~\ref{lem:lem2}] 
	We prove the guarantee in each interval and then combine all the intervals to get the overall approximation guarantee.
	
	$\bullet$ \textbf{Consider an interval $[L_i, U_i]$.} Let define $\mathcal{B}_{S}$ the collection of \textit{bad} sets of $k$ nodes in the sense that the coverage of each set is less than $(1-1/e-\epsilon)$ the optimal coverage of $k$ nodes, i.e. $\forall S \in \mathcal{B}_{S}, \Covw(S) < (1-1/e-\epsilon) \OPTk{}$ or similarly $\forall S \in \mathcal{B}_{S}, \mu_{S} < (1-1/e-\epsilon) \mu_{S^*}$. We show that any bad set in $\mathcal{B}_{S}$ is returned with probability at most $\frac{4}{p}$ by considering two events:
	\begin{itemize}
		\item[(E1)] $\max_{L_i \leq j \leq U_i} \frac{\sum_{i = 1}^j X^{(i)}_{S^*}}{j} \leq (1-\frac{\epsilon-\epsilon_2(1+\alpha)}{1-1/e}) \mu_{S^*}$ and the returned solution $\hat S \in \mathcal{B}_{S}$.
		\item[(E2)]  $\max_{L_i \leq j \leq U_i} \frac{\sum_{i = 1}^j X^{(i)}_{S^*}}{j} > (1-\frac{\epsilon-\epsilon_2(1+\alpha)}{1-1/e}) \mu_{S^*}$ and the returned solution $\hat S \in \mathcal{B}_{S}$.
	\end{itemize}
	Thus, the probability of returning a set $S \in \mathcal{B}_S$ is $$\Pr[E1] + \Pr[E2].$$ We bound each of these events in the following.
	
	The probability of the first joint event (E1) is less than the probability of $\max_{L_i \leq j \leq U_i} \frac{\sum_{i = 1}^j X^{(i)}_{S^*}}{j} \leq (1-\frac{\epsilon-\epsilon_2}{1-1/e}) \mu_{S^*}$ which is bounded by $\frac{2}{p}$ as shown below.
	
	Apply Corollary~\ref{cor:n_samples} with $\mu = \mu_{S^*} = \OPTk$ and $N = U_i \geq T^* \geq N(\frac{\epsilon/(1+\alpha)-\epsilon_2}{1-1/e}, \frac{1}{p}, \mu_{S^*})$,
	\vspace{-0.05in}
	\begin{align}
	\Pr \Big[ \max_{1 \leq j \leq U_i} \Big| \sum_{i = 1}^j(X^{(i)}_{S^*} - \mu_{S^*}) \Big| \geq \frac{\epsilon/(1+\alpha)-\epsilon_2}{1-1/e} \mu_{S^*} U_i \Big] \leq \frac{2}{p} \nonumber
	\end{align}
	Consider the sub-interval $j \in [L_i, U_i]$ and note that $\frac{U_i}{j} \leq \frac{U_i}{L_i} = (1+\alpha)$ and rearrange the terms, we obtain:
	\vspace{-0.05in}
	\begin{align}
	%			\label{eq:opt_bound}
	\Pr[E1] & = \Pr \Big[ \max_{L_i \leq j \leq U_i} \frac{\sum_{i = 1}^j X^{(i)}_{S^*}}{j} \leq (1-\frac{\epsilon-\epsilon_2(1+\alpha)}{1-1/e}) \mu_{S^*} \Big] \nonumber \\
	& \leq \frac{2}{p}. \nonumber
	\end{align}
	\vspace{-0.1in}
	
	The second joint event (E2) implies:
	\begin{itemize}
		\item[(1)] $\sum_{i=1}^{T} X^{(i)}_{\hat S} \geq (1-1/e)\sum_{i=1}^{T} X^{(i)}_{S^*}$ (due to Lemma~\ref{lem:anyz})
		\item[(2)] $\mu_{\hat S} < (1-1/e-\epsilon) \mu_{S^*}$
		\item[(3)] $\frac{\sum_{i = 1}^T X^{(i)}_{S^*}}{T} > (1-\frac{\epsilon-\epsilon_2(1+\alpha)}{1-1/e}) \mu_{S^*}$
	\end{itemize} 
	From the three above inequalities, we derive
	\vspace{-0.05in}
	\begin{align}
	\frac{\sum_{i = 1}^T X^{(i)}_{\hat S}}{T} & \geq (1-1/e)\frac{\sum_{i = 1}^T X^{(i)}_{S^*}}{T} \nonumber \\
	& > (1-1/e) (1-\frac{\epsilon-\epsilon_2(1+\alpha)}{1-1/e}) \mu_{S^*} \nonumber \\
	& > (1-1/e-\epsilon)\mu_{S^*} + \epsilon_2(1+\alpha)\mu_{S^*} \nonumber \\
	& > \mu_{\hat S} + \epsilon_2(1+\alpha)\mu_{S^*}.
	%			(1-1/e-\epsilon) \frac{\mu_{\hat S}}{1-1/e-\epsilon} + \epsilon_2 \mu_{S^*} \nonumber \\
	%			\Leftrightarrow & \max_{L_i \leq j \leq U_i} \frac{\sum_{i = 1}^j X^{(i)}_{\hat S}}{j} - \mu_{\hat S} \geq \epsilon_2 \mu_{S^*}
	\end{align}
	\vspace{-0.15in}
	
	\noindent Hence, the probability of (E2) is bounded by
	\vspace{-0.05in}
	\begin{align}
	%			\label{eq:e2_prob}
	\Pr[E2] \leq \Pr \Big [\max_{L_i \leq j \leq U_i} \frac{\sum_{i = 1}^j X^{(i)}_{\hat S}}{j} > \mu_{\hat S} + (1+\alpha)\epsilon_2 \mu_{S^*} \Big ]. \nonumber
	\end{align}
	\vspace{-0.1in}
	
	\noindent Apply again Corollary~\ref{cor:n_samples} on a set $S$ with $U_i \geq T^* \geq N((1+\alpha)\epsilon_2, \frac{1}{p}\frac{1}{{n \choose k}}, \mu_{S^*})$ hyperedges and consider the sub-interval $[L_1,U_1]$,
	\vspace{-0.08in}
	\begin{align}
	& \Pr \Big[ \max_{L_i \leq j \leq U_i} \Big| \sum_{i = 1}^j(X^{(i)}_{S} - \mu_{S}) \Big| \geq (1+\alpha)\epsilon_2 \mu_{S^*} U_i \Big] \leq \frac{2}{p}\frac{1}{{n \choose k}} \nonumber \\
	& \Rightarrow \text{ } \Pr \Big[ \max_{L_i \leq j \leq U_i} \frac{\sum_{i = 1}^j X^{(i)}_{\hat S}}{j} - \mu_{\hat S} \geq \epsilon_2 (1+\alpha) \mu_{S^*} \Big] \leq \frac{2}{p}. \nonumber
	\end{align}
	\vspace{-0.1in}
	
	Therefore, we obtain
	\vspace{-0.05in}
	\begin{align}
	\Pr[E2] \leq \frac{2}{p}.
	\end{align}
	\vspace{-0.15in}
	
	Combine the probabilities of (E1) and (E2), we find that the probability of returning a bad set is at most $\frac{4}{p}$ when $T \in [L_i, U_i]$.
	%		\vspace{-0.05in}
	%		\begin{align}
	%			\Pr\left[ T \in [L_i, U_i] \textsf{ and } \hat S \in \mathcal{B}_S \right] \leq 4\left(\frac{1}{p}\right)^{1+\alpha i}.
	%		\end{align}
	%		\vspace{-0.1in}
	
	$\bullet$ \textbf{Consider all $[L_i, U_i]$, $\forall i = 1, \dots, \lc\log_{1+\alpha}c\rc$}. The probability of returning a bad set in any of the intervals $[L_i, U_i]$ is bounded by
	\begin{align}
	\lc\log_{1+\alpha}c\rc \frac{4}{p},
	\end{align}
	because there are $\lc\log_{1+\alpha}c\rc$ such intervals.
\end{proof}

\subsection{Proof of Lemma~\ref{lem:lu_bounds}}
We state the following general lemma and apply it in a straightforward manner to derive the results in Lemma~\ref{lem:lu_bounds}.
\vspace{-0.05in}
\begin{Lemma}
	\label{lem:bounds}
	Let $X_1, X_2, \dots, X_N$ be weakly dependent Bernoulli random variables satisfying $\forall i=1..N,$
	\[
	E[X_i| X_1, X_2,\ldots,X_{i-1}] =E[X_i^2| X_1, X_2,\ldots,X_{i-1}]=\mu.\]
	For $\delta \in (0, 1)$, define $c_i = \ln (\frac{1}{\delta})/i, a = \frac{i}{N}$ and $\hat \mu_i = \frac{1}{j}\sum_{j = 1}^{i} X_j$,
	\vspace{-0.05in}
	\begin{align}
		\label{eq:lbound}
		\Pr\left[ \mu \geq f_L(i, \hat \mu_i, \delta, N),\  \forall i=1..N \right]\geq 1-\delta, \\
		\label{eq:ubound}
		\Pr\left[ \mu \leq  f_U(i, \hat \mu_i, \delta, N),\  \forall i=1..N \right]\geq 1-\delta,
	\end{align}
	where
	\begin{align}
		\label{eq:lower_bound}
		&f_L(i, \hat \mu_i, \delta,N)= \min \Bigg\{ \hat \mu_i +\frac{(\hat \mu_i-1)  c_i}{3 - c_i}, \\
		&\Resize{8cm}{\frac{3c_i + 3a\hat \mu_i - ac_i (\hat \mu_i+1)-\sqrt{c_i^2 (3+a(\hat \mu_i-1))^2 +18 a (1-\hat \mu_i)\hat \mu_i}}{c_i(6-2a)+3a}} \Bigg\}, \nonumber
	\end{align}
	\vspace{-0.15in}
	
	\noindent and,

	\vspace{-0.1in}
	\begin{align}
		\label{eq:upper_bound}
		&f_U(i, \hat \mu_i, \delta,N)= \max \Bigg\{ \hat \mu_i +\frac{(i-\hat \mu_i)  c_i}{3 + c_i}, \\
		& \Resize{8cm}{\frac{3c_i + 3a \hat \mu_i + a c_i (1+\hat \mu_i)+\sqrt{c_i^2(3+a(1-\hat \mu_i))^2 +18a(1-\hat \mu_i)\hat \mu_i}}{c_i (6+2a) + 3a}} \Bigg\}. \nonumber
	\end{align}
	% 	1) $\mu_X \geq f_l(S,\E_r),$ with probability at least $1-e^{-u}$ where 
	% 	\begin{align}
	% 		f_l(S,\E_r) & = \min \left\{ \frac{3a - c}{3 + c}, \right. \nonumber \\
	% 		& \left. \frac{3a + 2c-ac-\sqrt{c^2(a+2)^2+18ac(1-a)}}{4c+3} \right \},
	% 	\end{align}
	% 	\indent $a = \hat \mu_X$ and $c=\frac{u}{n}$; and,
	% 	\vspace{0.05in}
	
	% 	2) $\mu_X \leq f_u(S,\E_r)$ with probability at least $1-e^{-u}$ where
	% 	\begin{align}
	% 		f_u(S,\E_r) & =  \max \left\{ \frac{3a + c}{3 - c}, \right. \nonumber \\
	% 		& \left. \frac{3a + 4c+ac+\sqrt{c^2(4-a)^2+18ac(1-a)}}{8c+3} \right \}.
	% 	\end{align}
\end{Lemma}

	From the bound in Lemma~\ref{lem:central_bound}, we derive the following:
	
	\begin{lemma}
		\label{lem:t_bound}
		Given Bernouli random variables $X_1, \ldots, X_N$ with mean  $E[X_i] = \mu \leq 1/2$, let $\hat \mu_i=\frac{1}{i}\sum_{j=1}^i X_j$, $\delta\in(0,1)$ and fixed $N$. Denote $c = \ln (\frac{1}{\delta})/N$ and $t = \frac{1}{3}c (1-\mu)  + \sqrt{\frac{1}{9}c^2(1-\mu)^2 + 2\mu(1-\mu)c}$,
		\vspace{-0.05in}
		\begin{align}
		\Pr\left[ \hat \mu_i- \mu \leq \frac{N}{i}t,\  \forall i=1..N \right]\geq 1-\delta,
		\end{align}
		\vspace{-0.2in}
		
		and,
		\vspace{-0.05in}
		\begin{align}
		\Pr\left[ \hat \mu_i- \mu \geq - \frac{N}{i}t,\  \forall i=1..N \right]\geq 1-\delta.
		\end{align}
	\end{lemma}
	%\begin{proof}	
	\vspace{-0.23in}
	
\balance	
\begin{proof}
	\textit{Deriveing Lower-bound function:} By Lemma~\ref{lem:t_bound}, we have
	\vspace{-0.08in}
	\begin{align}
	\Pr\left[ \hat \mu_i- \mu \leq \frac{N}{i} t,\  \forall i=1..N \right]\geq 1-\delta,
	\end{align}
	\vspace{-0.15in}
	
	where
	\vspace{-0.05in}
	\begin{align}
	t = \frac{1}{3}c (1-\mu)  + \sqrt{\frac{1}{9}c^2(1-\mu)^2 + 2\mu(1-\mu)c}.
	\end{align}
	\vspace{-0.15in}
	
	\noindent Therefore, with probability at least $1-\delta$, $\forall i = 1..N$
	\vspace{-0.05in}
	\begin{align}
	& \hat \mu_i - \mu \leq \frac{1}{3}c_i(1-\mu) + \sqrt{\frac{c_i^2}{9}(1-\mu)^2 + 2\mu(1-\mu)c_i\frac{N}{i}} \nonumber \\
	\Leftrightarrow \text{ }& \hat \mu_i - \mu - \frac{1}{3}c_i(1-\mu) \leq \sqrt{\frac{c_i^2}{9}(1-\mu)^2 + 2\mu(1-\mu)c_i\frac{N}{i}} \nonumber
	\end{align}
	\vspace{-0.15in}
	
	\noindent From the above inequality, $\forall i =1..N$, consider two cases:
	\begin{itemize}
		\item $\hat \mu_i - \mu -\frac{1}{3}c_i(1-\mu) < 0$, then $\mu > \frac{\hat \mu_i - c_i/3}{1-c_i/3}$
		\item $\hat \mu_i - \mu -\frac{1}{3}c_i(1-\mu) \geq 0$, then taking the square of two sides, we have,
		\vspace{-0.05in}
		\begin{align}
		\left(\hat \mu_i - \mu -\frac{c_i}{3}(1-\mu)\right)^2 \leq \frac{c_i^2}{9}(1-\mu)^2 + 2c_i\frac{N}{i}\mu(1-\mu).
		\end{align}
		Solve the above quadratic inequality for $\mu$, we obtain another lower-bound for $\mu$.
	\end{itemize}
	%		\begin{align}
	%			\mu_X \geq \frac{(3a + 2c-ac)-\sqrt{c^2(a+2)^2+18ac(1-a)}}{4c+3},
	%		\end{align}
	%		where $a = \hat \mu_X$.
	%	\end{itemize}
	%Thus, taking both cases into consideration, we obtain the lower bound function for $\mu$.
	%\begin{align}
	%	\label{eq:lower_bound}
	%	f_l(S,\E_r) & = \min \left\{ \frac{3a - c}{3 + c}, \right. \nonumber \\
	% 	& \left. \frac{3a + 2c-ac-\sqrt{c^2(a+2)^2+18ac(1-a)}}{4c+3} \right \},
	%\end{align}
	% 	where $a = \hat \mu_X$ and $c=\frac{u}{n}$.
	
	\textit{Deriving Upper-bound function.} By Lemma~\ref{lem:t_bound},
	\vspace{-0.05in}
	\begin{align}
	\Pr\left[ \hat \mu_i- \mu \geq - \frac{n}{i}t,\  \forall i=1..N \right]\geq 1-\delta.
	\end{align}
	\vspace{-0.15in}
	
	Follow the similar steps as in deriving the lower-bound function above, we obtain the upper-bound function in the lemma.
	%\begin{align}
	%	\label{eq:upper_bound}
	%	f_u(S,\E_r) & =  \max \left\{ \frac{3a + c}{3 - c}, \right. \nonumber \\
	%	& \left. \frac{3a + 4c+ac+\sqrt{c^2(4-a)^2+18ac(1-a)}}{8c+3} \right \}.
	%\end{align}
\end{proof}
\vspace{-0.12in}
The lower-bound function $f_L(N_j, d_c/N_j, \delta', N_j)$ is directly obtained from $f_L(i, \hat \mu_i, \delta, N)$. Notice that we can replace $\hat \mu_i$ in the upper-bound $f_U(i, \hat \mu_i, \delta, N)$ by a larger number and still obtain an upper-bound with the same probability guarantee. Thus, the upper-bound $f_U \big(T_z, z/T_z, \delta', \lc (1+\beta)^{t_u} \rc \big)$ is derived from $f_U(i, \hat \mu_i, \delta, N)$ by using the larger value $z/T$ for the coverage of optimal solution.
\vspace{-0.08in}
\subsection{Proof of Theorem~\ref{lem:dta_quality}}
%\begin{proof}
	\dta{} returns either $S_z$ when the test $\rho_S \geq 1-1/e-\epsilon$ succeeds (Lines~9, 10) or $S_{z^*}$ (Line~13). We show that the corresponding returned solutions are $(1-1/e-\epsilon)$-approximate with probabilities $1-4\delta/3$ and $1-\delta$, respectively. Then the lemma follows by adding these probabilities.
	
	\emph{Return $S_z$ when the test $\rho_S \geq 1-1/e-\epsilon$ succeeds.} The $(1-1/e-\epsilon)$ approximation ratio is directly obtained by Lemma~\ref{lem:lu_bounds}. The probability of this guarantee is computed by taking the sum of the two probabilities in Lemma~\ref{lem:lu_bounds} times the number of times that $\rho_S$ is computed. By Lemma~\ref{lem:lu_bounds}, each bound $\rho_S$ holds with prob. $1-2\delta' = 1-\frac{\delta}{\log_2(z^*) \log_{1+\beta}(cT^*)}$ since the lower and upper bounds returned by $f_L(.)$ and $f_U(.)$ individually hold with probability $1-\delta'$. We show that the total number of times that $f_L(.)$ and $f_U(.)$ are computed is at most $\log_2(z^*) \log_{1+\beta}(cT^*)$ with probability $1-\delta/3$. Then, the accumulated probability is given by $1-\frac{\delta}{\log_2(z^*) \log_{1+\beta}(cT^*)} \cdot \log_2(z^*) \log_{1+\beta}(cT^*) - \delta/3 = 1-4\delta/3$. Thus, it is sufficient to show that the total number of times that $f_L(.)$ and $f_U(.)$ are computed is at most $\log_2(z^*) \log_{1+\beta}(cT^*)$.
	
	For $z = z^*/2^i, i \geq 1$, similar to the proof of Eq.~(\ref{eq:u_bound}), \tpa{} with threshold $z$ acquires $T \geq cT^*$ hyperedges with probability at most
	\vspace{-0.12in}
	\begin{align}
		\Pr \Big[\sum_{i = 1}^{T} X^{(i)}_{S^*} \leq z \Big] \leq 2 \Big (\frac{1}{p} \Big )^{2^i} < 2 \Big (\frac{1}{p} \Big )^{i}.
	\end{align}
	\vspace{-0.12in}
	
	Thus, the probability that \tpa{} with $z$ samples at most $cT^*$ hyperedges, hence, there are at most $\log_{1+\beta}(cT^*)$ times that $f_L(.)$ and $f_U(.)$ are computed, is $1-2 (\frac{1}{p})^{i}$. Then, accumulating this probability over all \tpa{} calls with $z \in \Big \{ \frac{z^*}{2^{i_0}}, \frac{z^*}{2^{i_0-1}}, \dots, z^* \Big \}$, we obtain the probability of \dta{} having at most $\log_2(z^*) \log_{1+\beta}(cT^*)$ computation rounds of $f_L(.)$ and $f_U(.)$ to be at least
	\vspace{-0.05in}
	\begin{align}
		1-\sum_{i = 1}^{\log_2(z^*)} 2 \Big (\frac{1}{p} \Big )^{i} \geq 1 - 2\frac{1}{p}\frac{1}{1-1/p} \geq 1 - \frac{\delta}{3},
	\end{align}
	where $p$ was replaced by its definition in Eq.~(\ref{eq:p}).
	
%	Furthermore, there are $\log_2(z^*)$ trial threshold $z$ and for each of them, there are at most $\log_{1+\beta}(cT^*)$ check points. Thus, we have in total $\log_2(z^*) \log_{1+\beta}(cT^*)$ check points, hence, the probability that all of them are correct is $1-\delta$. Since \dta{} stops at one of those points and $\rho_S \leq 1-1/e-\epsilon$, the returned solution is $(1-1/e-\epsilon)$-approximate with prob. $1-\delta$.
	
	\emph{\dta{} returns $S_{z^*}$.} By Theorem~\ref{theo:tpa}, \tpa{} with $z = z^*$ finds an $(1-1/e-\epsilon)$-approximate solution with probability $1-\delta$.